\documentclass[ssy,preprint]{imsart}

\usepackage{amsmath,amsthm}
\usepackage{mathrsfs,amssymb}
\usepackage{amsfonts,url,bbm,enumerate}
\usepackage{graphicx, subfigure, color}
\usepackage[draft]{fixme}
\usepackage{enumitem}

\newtheorem{theorem}{Theorem}
\newtheorem{lemma}{Lemma}
\newtheorem{proposition}{Proposition}
\newtheorem{corollary}{Corollary}

\newtheorem{claim}{Claim}
\newtheorem{definition}{Definition}

\newtheorem{assumption}{Assumption}

\RequirePackage[OT1]{fontenc}
\RequirePackage[numbers]{natbib}
\RequirePackage[colorlinks,citecolor=blue,urlcolor=blue]{hyperref}

\newcommand{\hide}[1]{}
\renewcommand{\hat}[1]{\widehat{#1}}
\renewcommand{\bar}[1]{\overline{#1}}
\renewcommand{\tilde}[1]{\widetilde{#1}}
\renewcommand{\S}{\mathcal{S}}
\newcommand{\E}[1]{\mathbb{E}\left\{#1\right\}} 
\newcommand{\prob}[1]{\mathrm{P}\left(#1\right)} 
\newcommand{\probb}[1]{\mathrm{P}\big(#1\big)} 
\newcommand{\probB}[1]{\mathrm{P}\Big(#1\Big)} 
\newcommand{\probbb}[1]{\mathrm{P}\bigg(#1\bigg)} 
\newcommand{\probbB}[1]{\mathrm{P}\Bigg(#1\Bigg)} 

\newcommand{\R}{\mathbb{R}}
\newcommand{\N}{\mathbb{N}}
\newcommand{\Z}{\mathbb{Z}}

\newcommand{\Ltwo}[1]{\big\lVert #1 \big\rVert}

\newcommand{\divconst}{{C}}  

\newcommand{\conv}{\mathrm{Conv}}

\newcommand{\region}{R}

\newcommand{\ptraj}{X}

\newcommand{\ff}{\zeta}  

\newcommand{\capregion}{\mathcal{C}}
\newcommand{\net}{N}
\newcommand{\qhat}{{\hat q}}

\newcommand{\absorbconst}{\alpha}

\newcommand{\RM}{R}
\newcommand{\I}{\mathcal{I}}
\newcommand{\Prt}{U}
\newcommand{\tf}{ { \lfloor t\rfloor } }
\newcommand{\tr}{\big(t-\lfloor t\rfloor\big)}
\newcommand{\ones}{\mathbbm{1}}
\newcommand{\diag}{\mathrm{diag}}
\newcommand{\ltwo}[1]{\lVert #1 \rVert}

\newcommand{\argmax}[1]{\underset{#1}{\operatorname{argmax}}}

\newcommand{\ee}{\mathcal{E}}

\long\def\red#1{{\color{red}#1}}
\usepackage[normalem]{ulem}

\long\def\green#1{{\color{green}#1}}

\long\def\pur#1{{\color[rgb]{.8,0,.8}#1}}
\long\def\oli#1{{\color[rgb]{0,.8,.8}[#1]}}
\newcommand\redsout{\bgroup\markoverwith{\textcolor{red}{\rule[0.5ex]{2pt}{.5pt}}}\ULon}
\long\def\dela#1{{\color{blue}\redsout{[#1]}}}

\long\def\comm#1{{\color{green} [#1]}}
\long\def\old#1{}

\begin{document}

\begin{frontmatter}
\begin{aug}
\title{Fluctuation Bounds for the Max-Weight Policy, with Applications to State Space Collapse}

\author{\fnms{Arsalan} \snm{Sharifnassab}\thanksref{m1,t1}\ead[label=e2]{asharif@mit.edu}},
\author{\fnms{John N.} \snm{Tsitsiklis}\thanksref{m2}\ead[label=e3]{jnt@mit.edu}},
\author{\fnms{S. Jamaloddin} \snm{Golestani}\thanksref{m1}\ead[label=e1]{golestani@sharif.edu}},


\thankstext{m1}{A. Sharifnassab and S. J. Golestani are with the Department of Electrical Engineering, Sharif University of Technology, Tehran, Iran; email: asharif@mit.edu, golestani@sharif.edu}
\thankstext{m2}{J. N. Tsitsiklis is with the Laboratory for Information and Decision Systems, Electrical Engineering and Computer Science Department, MIT, Cambridge MA, 02139, USA; email: jnt@mit.edu}
\thankstext{t1}{This work was partially done while A.\ Sharifnassab was a visiting student at the Laboratory for Information and Decision Systems, MIT, Cambridge MA, 02139, USA.}


\end{aug}

\vspace{3mm}
\today
\vspace{3mm}

\begin{abstract}
We consider a multi-hop switched network operating under a Max-Weight (MW) scheduling policy, and show that the distance between the queue length process and a fluid solution remains bounded by a constant multiple of the deviation of the cumulative arrival process from its average. 
We then exploit this  result to prove matching upper and lower bounds for the time scale over which  additive state space collapse (SSC) takes place.  
This implies, as two special cases, 
an additive SSC result in diffusion scaling under non-Markovian arrivals 
 and, for the case of i.i.d.\  arrivals, an additive SSC result over an exponential time scale.
\end{abstract}
\end{frontmatter}
 \tableofcontents

\section{\bf Introduction}
The subject of this paper is a new line of analysis of the Maximum Weight (MW) scheduling policy for single-hop and multi-hop networks. The main ingredient is a purely deterministic qualitative property of the queue dynamics: the trajectory followed by the queue vector under a MW policy tracks the trajectory of an associated deterministic  fluid model, within a constant multiple of the cumulative fluctuation of the arrival processes. With this property at hand, it is then a conceptually simple matter to translate concentration properties of the arrival processes to concentration properties for the queue vector. As a consequence, we can obtain:
\begin{itemize}
\item[(a)] New, simple derivations of existing results on the convergence to a fluid solution/trajectory and on state space collapse (SSC).
\item[(b)] Stronger versions of existing SSC results, involving more general arrival processes, and tighter concentration bounds.
\item[(c)] An approach to obtaining new results that would seem rather difficult to establish with existing methods. 
\end{itemize}
The core of our approach is the trajectory tracking result mentioned above. The latter is in turn an adaptation of a similar result established in \cite{AlTG19sensitivity}, for a general class of continuous-time hybrid systems that move along the sudifferential of a piecewise linear convex potential function with finitely many pieces; other than an additional restriction to the positive orthant, a continuous-time variant of the MW dynamics turns out to be exactly of this type. 
However, a fair amount of additional work is needed to translate the general result to the standard, discrete-time, MW setting; cf.\ Theorem~\ref{th:main sensitivity} and its proof.

\subsection{General Background}\label{s:gb}
We consider a multi-hop switched network with fixed routing, such as those arising in wireless networks \cite{LinSS06} or switch fabrics \cite{JiAS09}. 
The network operates in discrete time, and is driven by jobs (or packets) that  arrive according to a stochastic, deterministic, or  adversarial process. 
There is a scheduler which, at each time step, selects one of finitely many possible service vectors. These service vectors can be fairly arbitrary, reflecting interdependence constraints between different servers, e.g., interference constraints in the context of wireless networks.

We focus on the popular MW scheduling policy \cite{TassE92}, which operates 
as follows.
At any time step, a MW policy associates to each queue a weight proportional to its length, and selects a service vector that maximizes  the total weighted service.
MW policies are
known to have a number of attractive properties such as maximal throughput \cite{TassE92,EryiS05,GeorNT06}.
  In addition, under certain conditions, e.g., a resource pooling assumption, they minimize the  workload in the heavy traffic regime \cite{Stol04}. 
On the other hand, the queue size dynamics, under MW policies, are quite complex, and a detailed analysis is  difficult. 

A common way of reducing the complexity of the analysis involves a fluid approximation, also known as a fluid model.
The fluid model relies on two simplifications that lead to a description in terms of a set of differential equations (cf. Subsection \ref{subsec:fluid model}):  (a) the dynamics evolve in continuous --- rather than discrete --- time, and (b)  the arrival process is replaced by  a constant flow with the same average. 
The fluid model underlies
a general technique for dealing with discrete-time networks: approximate the queue lengths by fluid solutions and then analyze the fluid model. 
This approach has proved useful in the study of the MW dynamics, leading to results on stability (\cite{DaiP00,AndrKRSV04}), SSC (\cite{Stol04,ShahW06,ShahW12,DaiL05})
and delay stability under heavy tailed arrivals  (\cite{MarkMT16,MarkMT18}).
A key ingredient behind such results is an understanding of the accuracy with which fluid solutions approximate the original queue length processes; this paper contributes to this understanding.

\subsection{State Space Collapse Literature}
A prominent  application of  fluid models is in establishing state space collapse (SSC), i.e., that in the  heavy traffic regime, the queue length process stays close to a low-dimensional set, for a long time, and with high probability.\footnote{We note here the important distinction between multiplicative and additive (or strong) state space collapse, which is discussed further in Section \ref{s:ssc-lit}. The literature review here is mostly about multiplicative state space collapse. 
}

Seminal SSC results for communication networks were given in the works of Reiman \cite{Reim84}, Bramson \cite{Bram98}, and Williams \cite{Will98}. 
Subsequently, several works \cite{Stol04,ShahW06,ShahW12,KangKLW09} 
followed the general framework  of Bramson \cite{Bram98} to prove SSC under different scheduling policies, including for the case of MW policies. 
The general approach involves splitting an $O(r^2)$-long interval into intervals of length $O(r)$, and then showing that the fluid-scaled processes (i.e., $\hat{q}(t)=\frac1r Q(\lfloor rt \rfloor)$) stay close to the fluid solutions in each one of these  smaller intervals. 
The SSC results then follow from  the property that the fluid solutions  are attracted to a low-dimensional set, called the set of invariant points. 

For  single-hop networks with Markovian arrivals operating under a generalization of the MW policy, SSC was proved in \cite{Stol04}. 
It was also shown, in \cite{Stol04},   as a consequence of SSC, that the \emph{workload process} converges to a reflected Brownian motion, and that every MW-$\alpha$
policy\footnote{For any given 
$\alpha>0$, 
the MW-$\alpha$ policy is an extension of the MW policy in which
the ``weight'' of queue $i$ is  proportional to  $Q_i^{\alpha}$, where $Q_i$ is the length of the 
queue at node $i$.}
with  $\alpha > 0$ 
minimizes this workload among all scheduling algorithms.
The results of \cite{Stol04} were extended to  multi-hop networks in \cite{DaiL08}, and to another generalization of MW policies in \cite{ShakSS04}. 
For multi-hop networks with non-Markovian arrivals operating under MW-$\alpha$, SSC under diffusion scaling was studied in \cite{ShahW12}.
Several works \cite{KangW12,ShahTZ10} then used the results of \cite{ShahW12} to provide diffusion approximations for the MW dynamics. 
Finally,  
SSC has also facilitated the study of the steady-state expectation of the number of jobs in a network \cite{EryiS12,MaguS16,MaguSY14,MaguBS16,XieL15,WangMTS18}. 


\subsection{Preview of Results}\label{s:preview}
Our approach to the analysis of MW policies relies on a bound on the distance of the queue length  processes from the fluid solutions, 
in terms of the fluctuations of the cumulative arrival processes. In more detail,
we consider a
queue length process $Q(\cdot)$, driven by an arrival process $A(\cdot)$ with average rate $\lambda$, and compare $Q(\cdot)$ with  a fluid solution $q(\cdot)$ driven by 
a steady arrival stream with the same rate $\lambda$, under the
same initial conditions $q(0)=Q(0)$. 

We already know that, under suitable scaling, the trajectories of the original discrete-time process remain close to the fluid solutions. Furthermore, the fluid model is well-known to be non-expansive\footnote{A dynamical system  is called non-expansive if for any 
two trajectories, $x(\cdot)$ and $y(\cdot)$,
 we have $\frac{d}{dt}\Ltwo{x(t)-y(t)}\le 0$.}
\cite{Subr10}.   
By combining these facts, it is quite plausible that one should be able to derive bounds of the form
\begin{equation}\label{eq:bad}
\Ltwo{Q(t)-q(t)}\, \le\, c\, +\, \sum_{\tau=0}^{t-1} \Ltwo{A(\tau)-\lambda},
\end{equation} 
where $A(t)$ is the vector of arrivals at each one of the queues at time $t$, and $c$ is a constant which is
 independent of $A(\cdot)$.
However, our goal is to derive a stronger bound, of the form
\begin{equation}\label{eq:intro bound}
\Ltwo{Q(t)-q(t)}\, \le\, c\, +\, \divconst \,\max_{k<t}\, \Big\| \sum_{\tau=0}^{k}\big(A(\tau)-\lambda\big) \Big\|,
\end{equation}
for some constants $c$ and $\divconst$, independent of $A(\cdot)$  and $\lambda$.
The bounds in Eqs.\ \eqref{eq:bad} and \eqref{eq:intro bound} are qualitatively different. 
Under common probabilistic assumptions, 
and with high probability, 
$\sum_{\tau=0}^{t-1} \Ltwo{A(\tau)-\lambda}$ grows at a rate of $t$, whereas $\max_{k<t}\, \Ltwo{\sum_{\tau=0}^{k}\big(A(\tau)-\lambda\big)}$ only grows as (roughly) $\sqrt{t}$.

The sensitivity bound \eqref{eq:intro bound} allows us to make several contributions to the study of the MW policy.
\begin{itemize}

\item[(a)] We obtain a very simple proof of the 
 convergence of fluid-scaled processes to fluid solutions; cf. Corollary~\ref{c:fluid}.

\item[(b)]
We  establish a strong  SSC result for the MW policy.
In particular, we derive an upper bound and a matching lower bound on the 
time scale over which additive SSC takes place; cf. Theorem~\ref{th:strong ssc}.
As a corollary, when the arrivals are i.i.d, we establish  SSC  
for the process  $\tilde{q}(t) = Q(\lfloor e^{\alpha r} t\rfloor)/r$, for some constant $\alpha$, i.e., over an exponentially long time scale; cf. Corollary~\ref{prop:ssc iid trans}.

\item[(c)] In another corollary, we 
 establish 
an additive SSC result in diffusion scaling and under non-Markovian arrivals, which strengthens the currently available diffusion scaling results under the MW policy in several respects;
see Section \ref{s:ssc-lit} for more details. 

\item[(d)] As will be reported elsewhere, the sensitivity bound \eqref{eq:intro bound} provides  tools that allow us to resolve an open problem from \cite{MarkMT18}, on the delay stability in the presence of heavy-tailed traffic.
\end{itemize}

On the technical side,
 the proof of the sensitivity bound \eqref{eq:intro bound}  exploits a similar bound from our earlier work \cite{AlTG19sensitivity} on the  sensitivity of a class of hybrid subgradient dynamical systems to fluctuations of external inputs or disturbances. 
The main challenges here 
concern the
 transition from discrete to continuous time, as well as the presence of  boundary conditions, as queue sizes are naturally constrained to be non-negative.
For the proof of our  SSC results, we  follow the general framework of Bramson \cite{Bram98}, while also taking advantage of the sensitivity bound \eqref{eq:intro bound}. 
We believe that  our tight
characterization of the time scale over which SSC holds would have been very difficult without the strong sensitivity bound  \eqref{eq:intro bound}.


\subsection{Outline}
The rest of the paper is organized as follows. In the next section, we describe the network model and our conventions, along with some background on fluid models and SSC. In Section \ref{sec:main}, we present our central result, which is an inequality of the form \eqref{eq:intro bound}; cf.\ Theorem \ref{th:main sensitivity}.
Then, in Section \ref{sec:ssc}, we present our SSC results. We provide the  proofs of our   results  in Sections \ref{sec:proof sensitivity} and \ref{sec:proof ssc trans}, while relegating some of the details to appendices, for improved readability.
Finally, in Section \ref{sec:discussion}, we offer some  concluding remarks and discuss possible extensions.

\section{\bf System Model and Preliminaries}\label{sec:model}
In this section, we list our notational conventions, define the network model that we will study, and go over the necessary background on fluid models and State Space Collapse (SSC).

\subsection{Notation and Conventions} \label{subsec:model notation}
We denote by $\R,\,\R_+,\,\R_{++},\, \Z$, $\Z_+$, and $\N$  the sets of real numbers, non-negative reals, positive reals, integers,  non-negative integers, and positive integers, respectively. 

A vector $v\in\R^n$, will always be treated as a column vector, with components
$v_i$, for $i=1,\ldots,n$. We use $v^T$ and $\lVert v\rVert_2$, 
to denote the transpose  and the Euclidean norm of $v$, respectively. 
For any two vectors $v$ and $u$ in $\R^n$,  the relation $v\preceq u$ 
indicates that $v_i\leq u_i$, for all $i$. 
Furthermore, we use $\min(v,u)$ to denote the componentwise minimum, i.e., the vector with components $\min(v_i,u_i)$. 
For a vector $\mu\in\R^n$ and a set $J$ of indices, we use $\sigma_{-J}(\mu)$ to denote the vector whose $i$th entry is equal to the $i$th entry of $\mu$ if $i\not\in J$, and is equal to zero if $i\in J$. Finally, we let $\ones_n$ be the $n$-dimensional vector with all components equal to 1.

The notation $\textrm{Conv}(\cdot)$  stands for the convex hull of a set of vectors in $\R^n$. 
Given a vector $v\in\R^n$ and a set $A\subseteq\R^n$, we let $v+A=\big\{v+x\,:\,x\in A \big\}$. We  use  $d({v}\,,\,A)$ to denote the Euclidean distance of $v$ from the set $A$. 
Furthermore, if $W$ is  an $n\times n$ matrix, we let $W A=\left\{Wx\,\big|\, x\in A   \right\}$ be the image of the set $A$ under the linear transformation associated with $W$.
Given a vector ${v}\in\R^n$, $\diag({v})$ denotes the $n\times n$ diagonal matrix with the entries of $v$ on its main diagonal.

Finally, for a function $f:\R\to\R$, and with a slight departure from standard conventions, we use either $\dot{f}(t)$ or $\frac{d^+}{dt}f(t)$ to   denote the right derivative of $f$ at $t$,  assuming that it exists.


\subsection{The Network Model and the MW Policy}
A discrete-time multi-hop network with fixed deterministic routing is specified by $n$ queues, 
a non-negative
 $n\times n$ 
 routing matrix $\RM$,  and a finite set $\S \subset \R_+^n$ of \emph{actions} (or service vectors) that correspond to the different schedules  that can be applied at any time. 
 
The input to a network is a collection of $n$ discrete-time, non-negative arrival processes, described by functions $A_{{i}}:\Z_+\to \R_+$, where $A_i(t)$ stands for the  workload that arrives to queue $i$ during the $t$th time slot. 
Whenever the arrival processes are ergodic stochastic processes, we 
define the \emph{arrival rate} vector $\lambda\in \R_+^n$ as the vector whose $i$th component is the  average of the process $A_i(\cdot)$.
We will use $Q_i(t)$ to denote the (always non-negative) workload at queue $i$ at time $t$, and $Q(t)$ to denote the corresponding workload vector. 
In the sequel, we will use the terms \emph{workload}, \emph{queue size}, and \emph{queue length}, interchangeably. 
The evolution of $Q(t)$ is determined by the particular policy used to operate the network.

Given a network and an arrival process $A(\cdot)$, the evolution of the queue lengths  
 is given by:
\begin{equation} \label{eq:network evolution rule}
Q(t+1) = Q(t)   +A(t) + (\RM-I) \min\big( \mu(t),Q(t)\big), \quad \forall\, t\in \Z_+,
\end{equation}
where $\mu(t)$ is the service vector chosen by the policy at time $t$, and as mentioned earlier, $\min\big( \mu(t),Q(t)\big)$ is to be interpreted componentwise. Equation~\eqref{eq:network evolution rule} corresponds to the situation where a time slot begins with a   queue vector $Q(t)$, and then a service vector $\mu(t)$ is  chosen and applied. Finally, the new arrivals $A(t)$ are recorded at the end of the time slot and contribute to the new queue vector $Q(t+1)$.

 Note that the routing matrix $R$ is deterministic, pre-specified, and is not affected by the queue sizes or the scheduling policy. 
Single-hop networks correspond to the special case where $\RM$ is the zero matrix.
More generally, 
the most common case (single-path routing) is one where 
 the routing matrix has entries in $\{0,1\}$, 
with at most one nonzero entry in each column, and where 
 the $ij$th entry being one indicates that 
 any work completed at queue $j$ is transferred to queue $i$ for further processing.
However, we allow for more general non-negative matrices $\RM$  because this additional freedom does not affect the main proofs, and also allows for a simpler treatment of  weighted MW  policies; see Lemma \ref{lem:wmw2mw}, in the proof of Theorem \ref{th:main sensitivity}.


The following assumption will be in effect
throughout the paper, and is naturally valid in typical application contexts. 
\begin{assumption}\label{assumption 1}
For any $\mu\in\S$, and any $i\in\{1,\ldots,n\}$, the set  $\S$ also contains the vector $\sigma_{-\{i\}}(\mu)$, i.e., the vector obtained by setting the $i$th component of $\mu$ to zero.
\end{assumption}

\hide{\green{[Side-QUESTION: Is it also valid for typical backpressure contexts?]}
\oli{The backpressure conterpart would allow for switching off arbitrary links. Such assumption is readily satisfied under common ``interference models'', such as indepedndent-set (or graph-based) interference model and SINR model. I guess the assumption is non-restrictive in the Backpressur contex, but I'm not sure if it is typical or common.}}

According to Assumption \ref{assumption 1}, if a certain service vector $\mu$ is allowed, it is also possible to follow $\mu$ at all queues other than queue $i$, while providing no service to queue $i$. In particular, the zero vector is always an element of $\S$. 
On the technical side, Assumption \ref{assumption 1} appears innocuous; however, it is indispensable for the proof technique used in this paper, and has also been made in earlier work (cf.\ Section 4.1 of \cite{Stol05primaldual} and Assumption 2.3 of \cite{ShahW12}).

We now proceed to define
 weighted Max-Weight (WMW) policies,
 which can be viewed as either a generalization of MW policies or as a special case of the broader class of MW-$f$ policies\footnote{A MW-$f$ policy is obtained by replacing $Q^TW$ in \eqref{eq:maximal} by $f(Q)$, where  $f:\R^n\to\R^n$ is a function in an appropriate class.}  considered in  \cite{EryiS05}, and backpressure-based utility maximization algorithms\footnote{Max-Weight is a special case, with the utility function equal to zero.} considered in \cite{Stol05primaldual,GeorNT06,Neel10}. 
We are given a multihop network with $n$ queues, as described above,
 along with a positive vector $w\in\R_{++}^n$, and the associated diagonal matrix  $W=\diag(w)$. For any $Q\in \R_+^n$, we let $\S_w(Q)$ be the set of maximizers of $ Q^T \,W\, \big(I-\RM\big)\mu$:
\begin{equation}\label{eq:maximal}
\S_w(Q) \triangleq \argmax{\mu\in\S}\,  Q^T W \big(I-\RM\big)\mu.
\end{equation}
A WMW policy associated with $w$ (or $w$-WMW, for short) chooses, at each time $t$, an arbitrary service vector $\mu{(t)}\in \S_w\big(Q(t)\big)$.\footnote{For a concrete example, if $\mu$ corresponds to serving only queue $j$, with unit service rate, and if work completed at queue $j$ is routed to queue $i$, the term $Q^T W(I-\RM)\mu$ is of the form $w_jQ_j-w_iQ_i$.}
A Max-Weight (MW) policy is a special case of a WMW policy, in which $w=\ones_n$. When dealing with MW policies, we drop
 the subscript $w$, and write $\S(Q)$ instead of $\S_{\ones_n}(Q)$.

Consider an ergodic and Markovian arrival process with  arrival rate vector $\lambda\in\R_+^n$, for which there exists some scheduling policy that 
stabilizes the network, i.e., results in a positive recurrent process. 
The  closure of the set of all such vectors $\lambda$ is called the \emph{capacity region}  
and is denoted by $\capregion$. 

We now record a fact that will be used later, in the proofs of Lemma \ref{lem:known lemma!} and Claim \ref{claim:I is low dim}.
\old{\footnote{To justify this fact, we can
follow the argument in Section 3.C of
\cite{TassE92}. (While \cite{TassE92} was phrased in terms of a restricted class  
of routing matrices $R$, the argument goes through in our setting as well.) In particular, 
we have that for
any $\lambda\in\capregion$, there exist vectors   $f$ and $c$ such that $f\preceq c\in\conv(\S)$   and  $\lambda=\big( I-\RM\big)f$.  
Assumption~\ref{assumption 1} then implies that $f\in\conv(\S)$, and as a result $\lambda\in\big( I-\RM\big)\conv(\S)$. 
\pur{It follows that $\capregion \subseteq \big( I-\RM\big) \conv(\S) $.}
\hide{$\lambda\preceq \big( I-\RM\big) \conv(\S)$. 
To conclude that $\lambda\in \big( I-\RM\big) \conv(\S)$, note that $\RM$ is a ``non-negative matrix'' (i.e., a matrix with non-negative entries), and by definition, $I-\RM$ is an ``M-matrix'' \cite{BermP94}. 
Then, $\big( I-\RM\big)^{-1}$ is also non-negative \cite{BermP94} (Theorem 2.4 of Chapter 6). Consider a $\nu\in \big( I-\RM\big) \conv(\S)$ that dominates $\lambda$, i.e., $\nu\succeq\lambda$. 
Then, $\big( I-\RM\big)^{-1}\lambda \preceq \big( I-\RM\big)^{-1} \nu \in\conv(\S)$. 
Assumption \ref{assumption 1} then implies that $\big( I-\RM\big)^{-1}\lambda  \in\conv(\S)$. 
As a result, $\lambda\in \big( I-\RM\big) \conv(\S)$, and \eqref{eq:capregion subset convex hull} follows.}}
\oli{Trying to give a simpler inline proof for \eqref{eq:capregion subset convex hull} and remove footnote 4. OK?}}
Fix some $\lambda\in\capregion$ in the capacity region and consider  a stabilizing policy.
We define $f_i$ as the averrage departure rate from queue $i$. 
Then, the flow conservation property $f=\lambda+\RM f$ implies that $\lambda =  \big( I-\RM\big) f$.
Moreover, following an argument similar to the one in Section 3.C of
\cite{TassE92}, there exists a  vector $c$ such that $f\preceq c\in\conv(\S)$.
Assumption~\ref{assumption 1} then implies that $f\in\conv(\S)$, and as a result $\lambda\in\big( I-\RM\big)\conv(\S)$.
In conclusion, 
\begin{equation} \label{eq:capregion subset convex hull}
\capregion \subseteq \big( I-\RM\big) \conv(\S) .
\end{equation}

A remarkable property of MW and WMW policies is that they are 
\emph{throughput optimal} in the  sense that for any $\lambda$ in the interior of 
$\capregion$, and any ergodic Markovian arrival process with average arrival rate vector $\lambda$, the resulting process is positive recurrent \cite{TassE92}. Similar throughput optimality results are available for extensions of MW, e.g., for the so-called $f$-MW policies 
\cite{EryiS05}.
\pur{}

\subsection{The Fluid Model} \label{subsec:fluid model}
The fluid model associated with the MW policy is a deterministic dynamical system that runs in continuous time, and in which the arrival stream is replaced by a steady ``fluid'' arrival stream with rate vector $\lambda$. 
We will be working with the following definition of the fluid model; somewhat different but equivalent definitions can be found  in \cite{ShahW12} and \cite{MarkMT18}.
\begin{definition}[Fluid  Solutions] \label{def:fluid net}
We are given an arrival rate vector $\lambda$ and an initial queue length vector $q(0)\in \R_+^n$.  A fluid model solution (or, simply, \emph{fluid solution}) is an absolutely continuous function $q:\R_+\to\R_+^n$ 
that together with a collection of functions $s_\mu:\R_+\to [0,1]$, for  $\mu\in\S$, and another function $y:\R_+\to\R_+^n$,  satisfies the following  relations, almost everywhere:
\begin{equation}
\dot{q}(t) = \lambda + \big(\RM-I\big) \left(\sum_{\mu\in\S}s_{\mu}(t)\mu \,-\, {y}(t)\right),\label{eq:fluid diff eq 1}
\end{equation}
\begin{equation}
\sum_{\mu\in\S}{s}_{\mu}(t) = 1,\label{eq:fluid diff eq 3}
\end{equation}
\begin{equation}
{y_i}(t)\le  \sum_{\mu\in\S}{s}_{\mu}(t)\mu_i,   \qquad  i{=}1,\ldots, n,\label{eq:fluid diff eq 4}
\end{equation}
\begin{equation}
\textrm{if }q_i(t)>0,\quad \textrm{then } {y}_i(t) = 0, \qquad   i{=}1,\ldots, n,\label{eq:fluid diff eq 5}
\end{equation}
\begin{equation}
\textrm{if } 
\hide{\sum_{i=1}^n \mu_i w_i q_i(t). <\max_{\nu\in\S}\,  \sum_{i=1}^n \nu_i w_i q_i(t),}
\mu \not\in \S_w\big(q(t)\big),
\quad \textrm{then } {s}_{\mu}(t)=0, \qquad \forall\ \mu\in \S. \label{eq:fluid diff eq 6}
\end{equation}
\end{definition}
It is known that for  any multi-hop network and any initial condition, a fluid solution 
always 
exists (cf.~Appendix~A of \cite{DaiL05} and Lemma~9 of \cite{MarkMT18}),
and is unique (cf. Lemma~10 of \cite{MarkMT18}), 
even though the corresponding $s_\mu(\cdot)$ and $y(\cdot)$ need not be unique. 
Moreover, for $q(0)\succeq 0$, \eqref{eq:fluid diff eq 1}--\eqref{eq:fluid diff eq 6} imply that $q(t)$ remains non-negative for all subsequent times $t$. 
Later on, in Proposition~\ref{prop:simulation fluid}, we will show that  fluid solutions 
admit an alternative description, as
 the trajectories of a related subgradient dynamical system.
 \footnote{This alternative description also explains why uniqueness holds, in contrast to the case of more general multiclass queueing networks; cf.~the last paragraph of the proof of Proposition~\ref{prop:simulation fluid}.}

We will be particularly interested in the set of invariant states of the fluid model, which, for any $\lambda$ in the capacity region, is defined by (cf.\ Theorem 5.4(iv) of \cite{ShahW12}) 
\begin{equation}\label{eq:def stable point network}
\I(\lambda)\triangleq \left\{q_0\in \R_+^n\,  \big|\, q(t)=q_0,\, \forall t,\, \textrm{ is a fluid solution} \right\}.
\end{equation}
Our notation is chosen to emphasize the dependence on $\lambda$ of the set of invariant states. 
We note that if $\lambda$ belongs to the interior of the capacity region, then
$\I(\lambda)$ is a singleton, equal to $\{0\}$. Thus, $\I(\lambda)$  can be non-trivial only if $\lambda$ lies on the boundary of $\capregion$.

We now record a scaling property of the set of fluid solutions.
\begin{lemma}\label{lem:fluid scale}
Consider a fluid solution $q(\cdot)$ and a constant $r>0$. Let $\hat{q}(t)= q(rt)/r$, for all $t\ge 0$. Then, $\hat{q}(\cdot)$ is also a fluid solution. 
\end{lemma}
\begin{proof}
Note that the set of maximizing schedules in Eq.\ \eqref{eq:maximal} does not change when we scale the queue vector by a positive constant. Therefore,
for any $t\ge 0$, 
\begin{equation}
\S\big(\hat{q}(t)\big) \,=\, \S\big(q(rt)/r\big)  \,=\, \S\big(q(rt)\big).
\end{equation}
Consider the functions $y(\cdot)$ and $s_\mu(\cdot)$  that together with $q(\cdot)$ satisfy the fluid model relations \eqref{eq:fluid diff eq 1}--\eqref{eq:fluid diff eq 6}. 
Let $\hat{y}(t)=y(rt)$ and $\hat{s}_\mu(t) = s_\mu({r}t)$, for all $t\ge0$ and all $\mu\in \S$. 
Then, it is easy to verify that $\hat{q}(\cdot)$, $\hat{y}(\cdot)$, and $\hat{s}_\mu(\cdot) $ also satisfy \eqref{eq:fluid diff eq 1}--\eqref{eq:fluid diff eq 6}. Therefore, $\hat{q}(\cdot)$ is also a fluid solution.
\end{proof}

Suppose that $\lambda$ is in the capacity region and that $q_0\in \I(\lambda)$, so that $q(t)=q_0$ is a fluid solution. Then,  Lemma~\ref{lem:fluid scale} implies that
 for any scalar 
$r>0$, 
$q(t)= q_0/r$ is also a fluid solution, 
and therefore  
$ q_0{/r}\in\I(\lambda)$. 
Furthermore, it is not hard to see that the identically zero function is also a fluid solution, so that $0\in \I(\lambda)$. 
We conclude that
$\I(\lambda)$ is a cone, i.e., 
\begin{equation} \label{eq:I is conic}
\alpha\I(\lambda) = \I(\lambda),\qquad {\forall \ \alpha > 0.}
\end{equation} 

The interest in fluid solutions stems from the fact that they provide approximations to suitably  scaled versions (i.e., under ``fluid scaling'') of the original process. We summarize here one such result, which is a special case of Theorem 4.3 in  \cite{ShahW12}; similar results are given in \cite{DaiL05} (Lemmas~4 and 5).

\begin{proposition}
Fix some $\lambda\in\R_+^n$, $T>0$, and $q_0\in\R_+^n$. 
Letting $r$ range over the positive integers, consider a sequence of arrival processes $A^r(\cdot)$ that satisfies
 \begin{equation}\label{eq:arate}
 \frac{1}{r}\, \max_{t\le rT}  \Ltwo{
\sum_{\tau=0}^t
 \big(A^r(\tau)-\lambda\big)}\,\xrightarrow[{\,r\to\infty\,}]{} 0,
 \end{equation}
 almost surely.
Let $Q^r(\cdot)$ be the process generated according to Eq.\ \eqref{eq:network evolution rule} when the arrival process is $A^r(\cdot)$ and the initial condition is $Q^r(0)=rq_0$.
We define the continuous-time scaled processes
 $\hat{q}^r(t) =Q^r(\lfloor rt\rfloor) /r$, and note that $\hat{q}^r(0)=q_0$, for all $r$.
 Finally, let $q(\cdot)$ be a fluid solution, under that particular vector $\lambda$, initialized with $q(0)=q_0$. 
Then,
 \begin{equation} \label{eq:th 4.3 shah fluid scaling}
\sup_{t\le T} \,\Ltwo{\hat{q}^r(t)- q(t)}\xrightarrow[{\,r\to\infty\,}]{}\, 0,
\end{equation}
 almost surely.
\end{proposition}
Condition \eqref{eq:arate} is typically satisfied under common probabilistic assumptions, e.g., when $A^r(\cdot)$ is an i.i.d.\ process with mean $\lambda$ and bounded domain, or more generally of exponential type.
Thus, loosely speaking, convergence of the arrival processes leads to convergence of the queue processes. 

\old{
\dela{We fix some $\lambda\in\R_+^n$ and $T>0$. 
Letting $r$ range over the positive integers, let us suppose that 
that we have a sequence of arrival processes $A^r(\cdot)$ that satisfies}
 \begin{equation}\label{eq:arate}
\dela{ \frac{1}{r}\, \max_{t\le rT}  \Ltwo{
\sum_{\tau=0}^t
 \big(A^r(\tau)-\lambda\big)}\,\xrightarrow[{\,r\to\infty\,}]{} 0,
 }
 \end{equation}
 \dela{
 almost surely.
(Condition \eqref{eq:arate} is typically satisfied under common probabilistic assumptions, e.g., when $A^r(\cdot)$ is an i.i.d.\ process with mean $\lambda$ and bounded domain, or more generally of exponential type.)
Let us also fix some $q_0\in\R_+^n$. Let $Q^r(\cdot)$ be the process generated according to Eq.\ \eqref{eq:network evolution rule} when the arrival process is $A^r(\cdot)$ and the initial condition is $Q^r(0)=rq_0$.
We define the scaled processes
 $\hat{q}^r(t) =Q^r(\lfloor rt\rfloor) /r$, and note that $\hat{q}^r(0)=q_0$, for all $r$.
 Finally, let $q(\cdot)$ be a fluid solution, under that particular vector $\lambda$, initialized with $q(0)=q_0$. We then have
 }
\begin{equation} \label{eq:th 4.3 shah fluid scaling}
\dela{\sup_{t\le T} \,\Ltwo{\hat{q}^r(t)- q(t)}\xrightarrow[{\,r\to\infty\,}]{}\, 0,}
\end{equation}
\dela{ almost surely.}

\dela{Thus, loosely speaking, convergence of the arrival processes leads to convergence of the queue processes. }

\oli{New paragraph.}}

As we shall see in Section~\ref{sec:main},  our results will allow for stronger statements; namely, we will show that the rate of convergence in Eq.\ \eqref{eq:arate} provides bounds on the rate of convergence to the fluid solution, in Eq.\ \eqref{eq:th 4.3 shah fluid scaling}; cf. Corollary~\ref{c:fluid}.

\subsection{State Space Collapse}\label{s:ssc-lit}
In this section, we discuss known results about State Space Collapse (SSC) under a  MW policy, thus setting the stage for a comparison with the results we will present in Section \ref{sec:ssc}.

We consider 
the heavy traffic regime, where the arrival rate vector 
gets arbitrarily close to some point $\lambda$ on the outer boundary of the capacity region.
 In this regime, the average queue lengths  typically tend to infinity, yet it is often
the case that the queue length vector stays close to the set of invariant states, $\I(\lambda)$. 
This phenomenon is called SSC, and has been studied extensively, mostly under the so-called diffusion scaling.
In this scaling, we start with a
sequence $Q^r(\cdot)$ of stochastic processes, indexed by $r\in\N$, 
and then proceed to study  a sequence of scaled processes $\qhat^{\,r}(\cdot)$, referred to as diffusion-scaled processes,  defined by
\begin{equation} \label{eq:def diffusion scale}
\qhat^{\,r}(t)\,=\, \frac{1}{r} Q^r\big( \lfloor r^2t \rfloor \big),\qquad  t \geq 0.
\end{equation}

The extent to which the queue length process stays close to the set of invariant states 
is in general determined by the magnitude of 
 the fluctuations  of the arrival process. It is therefore natural to start the analysis with some assumptions on these fluctuations. General SSC results, under the MW policy and some of its extensions, 
 were provided in \cite{ShahW12}, under the following assumption.\footnote{In our statement of the assumption, we modify the notation  of \cite{ShahW12},  interchanging the  roles of $z$ and $r$, to preserve consistency with the rest of this paper.}


\begin{assumption}[Assumption 2.5 of \cite{ShahW12}]\label{assumption 2.5}
Let $A^r(\cdot)$ be a sequence of arrival processes indexed by $r\in\N$. We assume that for each $r$, $A^r(\cdot)$ is stationary,\footnote{``Stationary'' means that the $A^r(t)$ have the same distribution  for all $t$, but without necessarily being independent.}
with  mean $\lambda^r$, and that $\lambda^r\to\lambda$ as $r\to\infty$. 
We furthermore
assume that there 
exists a sequence $\delta_r\in\R_+$ converging to $0$ as $r\to \infty$, such that
\begin{equation}\label{eq:th shah}
r \,\cdot \log^2 r \,\cdot\,  \prob{\max_{t\le r} \frac{1}{r} \Ltwo{\sum_{\tau= 0}^t\big( A^z(\tau)\,-\, \lambda^z\big)}\ge\delta_r}\xrightarrow[{\,r\to \infty\,}]{}0, \qquad  \textrm{uniformly in } z.
\end{equation}
\end{assumption}

Note that Assumption \ref{assumption 2.5} is quite general, not requiring the arrival processes to be i.i.d.\ or Markovian.
Theorem 7.1 of \cite{ShahW12}, slightly rephrased,\footnote{Our rephrasing consists of replacing the term denoted by $\Delta W\big(\qhat^r(t)\big)$ in 
\cite{ShahW12} by $\I(\lambda)$. This is legitimate, because $\Delta W\big(\qhat^r(t)\big)\in\I(\lambda)$ (cf.\ Theorem 5.4 (iv) in \cite{ShahW12}) and therefore
$d\big(\qhat^r(t)\,,\, \I(\lambda)\big) \leq d\big(\qhat^r(t)\,,\, \Delta W\big(\qhat^r(t)\big)\big)$.
}
establishes
 that for a network operating under a MW-$f$ policy (a generalization of WMW policies, and under certain conditions on $f$),  for any $T>0$, and under Assumption \ref{assumption 2.5}, the diffusion-scaled queue length processes $\qhat^r(\cdot)$ satisfy, for any $\delta>0$,
\begin{equation} \label{eq:th 7.1 of shah's paper}
\prob{\frac{\sup_{t\in  [0,T]}\,d\Big(\qhat^r(t)\,,\, \I(\lambda)\Big)}{\max\big(1, \,\sup_{t\in  [0,T]}\, \qhat^r(t)\big)} \,>\,\delta}  \,\xrightarrow[{\,r\to\infty\,}]{} \, 0,
\end{equation}
when $\lim_{r\to\infty}{\qhat^r}(0)=q_0$, for some $q_0\in\I(\lambda)$. 

The bound in \eqref{eq:th 7.1 of shah's paper} is referred to as  \emph{multiplicative} SSC.
Yet, there is a stronger notion, called \emph{additive} SSC, which involves a bound similar to \eqref{eq:th 7.1 of shah's paper}, but with the term $\max\big(\qhat^r(t),\,1\big)$ absent from the denumerator, and which is 
known to hold under i.i.d. arrivals.

\begin{theorem}[\cite{ShahTZ10} Theorem 7.7]  \label{conj:shah 7.2} 
Consider a network operating under a MW-$\alpha$ policy, with $\alpha\geq 1$, 
with  i.i.d. and uniformly bounded arrivals with rate $\lambda^r\to\lambda$, for some $\lambda\in\capregion$, and the associated  diffusion-scaled queue length processes $\qhat^r(\cdot)$. Assume that $\lim_{r\to\infty}{\qhat^r}(0)=q_0$, for some $q_0\in\I(\lambda)$. Then,\footnote{The result in \cite{ShahTZ10} assumed that $q_0=0$; however, the proof extends to the case of general $q_0$.}
for any $\delta>0$, 
\begin{equation}
\prob{\sup_{t\in  [0,T]}\,d\Big(\qhat^r(t)\,,\, \I(\lambda)\Big)>\delta}  \xrightarrow[{\,r\to\infty\,}]{} 0.
\end{equation}
\end{theorem}

Compared to the above literature, our results only apply to  the case where $\alpha=1$ (i.e., the MW policy), but allow for queue-dependent weights, so that the weight of queue $i$ is $w_i Q_i$. 
More crucially, 
our results (cf.\ Section \ref{sec:ssc} and Theorem \ref{th:strong ssc}, in particular):  
\begin{itemize}

\item[(a)]remain valid as long as $\lim_{r\to\infty}d\big({\qhat^r}(0),\I(\lambda)\big)=0$, which is a weaker condition than $\lim_{r\to\infty}{\qhat^r}(0)=q_0$, for a fixed $q_0\in\I(\lambda)$; 
\item[(b)] unlike \cite{ShahTZ10}, we do not require the arrival process to be i.i.d. or  bounded, as long as the arrival process has certain concentration properties. Furthermore, the concentration properties that we require
(cf.\ Definition \ref{def:f tailed}) are weaker than Assumption~\ref{assumption 2.5}, for the case of diffusion scaling (cf.\ Corollary \ref{prop:ssc diffusion});
\item[(c)] apply to scalings other than diffusion scaling, and include a converse result that characterizes the possible scalings for which additive SSC holds.
\end{itemize}

We finally note another related line of work which studies a property similar to SSC, namely, 
the extent to which the steady-state distribution is concentrated in a neighbourhood of the set of invariant points. In particular,  
\cite{MaguS15} and \cite{MaguBS16} have characterized the 
tail of the steady-state  distribution of 
the distance from the set of invariant points for the case of an input-queued switch.

\section{\bf Main Result: Sensitivity}\label{sec:main}
The backbone behind all of the results is  the following main theorem.
\begin{theorem}[Sensitivity of WMW policy] \label{th:main sensitivity}
For a network operating under a WMW policy, there exists a constant $\divconst$, to be referred to as the \emph{sensitivity constant},  that satisfies  the following. Consider an arrival process $A(\cdot)$ and the corresponding queue length process $Q(\cdot)$. Let $q(\cdot)$ be a fluid solution corresponding to some $\lambda\succeq 0$, and initialized with  $q(0)=Q(0)$. Then, for any $k\in \Z_+$,
\begin{equation} \label{eq:arrival bound}
\Ltwo{Q(k)-q(k)} \,\le\, \divconst\left(1 \,+\, \lVert \lambda\rVert\,+\, \max_{t< k} \Ltwo{\sum_{\tau=0}^{t} \big(A(\tau) \,-\, \lambda\big)} \right),
\end{equation}
\end{theorem}
Note that the result holds without having to assume that  $\lambda$ lies inside the capacity region.
The proof is given in Section \ref{sec:proof sensitivity}, and the key steps are as follows. We show that the study of WMW policies can be reduced to the study of MW policies. Furthermore, given
a network operating in discrete time under the MW policy, we introduce an associated continuous-time dynamical system, which we call the \emph{induced} dynamical system. 
Next, we  show that the fluid solutions and the queue length processes of the network can be viewed as  unperturbed and perturbed trajectories of the induced dynamical system, respectively. 
We finally argue that the induced dynamical system falls within the class of subgradient systems that were studied in \cite{AlTG19sensitivity}, 
 and apply the main result in that reference 
 to prove \eqref{eq:arrival bound}. The reductions that are developed in the course of the proof, may be of independent interest.

\subsection{Convergence to Fluid Model Solutions}
An immediate consequence of Theorem \ref{th:main sensitivity}, together with Lemma \ref{lem:fluid scale}, is a bound on the distance of the  fluid-scaled process $\hat{q}^r(t) =Q(\lfloor rt\rfloor)/r $ from a fluid solution $q(\cdot)$.

\begin{corollary}\label{c:fluid}
Consider a network operating under the WMW policy and let $C$ be the constant in Theorem \ref{th:main sensitivity}. Fix an arrival function $A(\cdot)$ and  some $q_0\in \R_+^n$. Let $Q^r(\cdot)$ be the process generated according to Eq.\ \eqref{eq:network evolution rule} when the arrival process is $A(\cdot)$ and the initial condition is $Q^r(0)=rq_0$. Let $\hat{q}^r(t) =Q^r(\lfloor rt\rfloor)/r $. 
Let $q(\cdot)$ be a fluid solution corresponding to some $\lambda\in\R_+^n$ and initialized at $q(0)=q_0$. Then, for any $T>0$,
\begin{equation}\label{eq:fluid scale close to FM}
\sup_{t\le T}\,\Ltwo{\hat{q}^r(t)-q(t)} \,\le\,  \frac{\divconst}{r} \, \max_{t< rT} \Ltwo{\sum_{\tau=0}^{t} \big(A(\tau) \,-\, \lambda\big)} \,+\, O\big(1/r\big).
\end{equation}
\end{corollary}
 
Corollary \ref{c:fluid}  strengthens
\eqref{eq:th 4.3 shah fluid scaling} significantly. Any statistical assumptions on the fluctuations of the arrival process $A(\cdot)$ readily yield concrete upper bounds on the distance of the original process from its fluid counterpart.

\section{\bf State Space Collapse}\label{sec:ssc}
In this section, we apply Theorem \ref{th:main sensitivity}   to establish a general additive SSC result; cf.\ Theorem \ref{th:strong ssc}.
\hide{We first discuss the transient regime in Subsection \ref{subsec:ssc trans}. We then  present our results for the steady state regime in Subsection \ref{subsec:ssc std}. [Lets Talk about the steady state regime. Not written here. Actually, in the iid arrivals, we don't really need Theorem \ref{th:main sensitivity} to give exponential bounds. It is too powerful. There are simpler proofs. So, I am not sure if we should mention the steady state regime as an application of Theorem \ref{th:main sensitivity}.]}
We then continue with some corollaries on exponential scaling or diffusion scaling.
Our approach can also be used to obtain results that apply in  steady-state.
However, we do not go into that latter topic because such results can also be proved using simpler, more direct methods, as in \cite{MaguS15} and \cite{MaguBS16}.

\hide{
\comm{Should we also reference this?\\ 
https://pubsonline.informs.org/doi/abs/10.1287/stsy.2018.0012}
\oli{Probably not. The above paper considers a continuous time markov chain on $M$ nodes, and each node falls in one of the $n$ possible states (with $M>>n$). As $M$ grows large, they are interested in the steady state distribution of the ``frequency of states'', that is the distribution of an $n$-dimensional vector $x$ whose $i$th element indicates the percentage of nodes being in state $i$. To do this, they employ mean field theory, and show that the equilibrium of a suitably defined determinstic continuous time dynamical system provides a decent approximation of the desired steady state distribution. However, its not clear how their model fits into network dynamics under MW policy. In the network, we have a fixed number of nodes and infinite state sapce, while they have infinite number of nodes and finite state space per node. Morover, we are interested in the steady state of the queue processes, not the frequency of states. In sum, establishing the usefulness of their results to steady state anaylsis of MW dynamics would require non-trivial arguments and steps, which I cannot see at this moment.}}

\subsection{Definitions and Preliminaries}

At the core of our proofs lies the following lemma,
which asserts that fluid solutions  are attracted to the set $\I(\lambda)$ of invariant states, which was defined in Eq.\ \eqref{eq:def stable point network}. The proof of the lemma is given in Appendix \ref{ap:a}.

\begin{lemma} [Attraction to the  Set of Invariant States] \label{lem:absorb const}
Consider a network operating under the MW policy and a vector $\lambda$ in its capacity region. There exists a constant $\absorbconst(\lambda)>0$ such that for any fluid solution $q(\cdot)$ associated with  $\lambda$, and any time $t$, 
$$q(t)\not\in \I(\lambda) \ \ \implies\ \  \frac{d^+}{dt}d\big(q(t)\,,\, \I(\lambda)\big)\le-\absorbconst(\lambda),$$ with this right-derivative being guaranteed to exist.
\end{lemma}
\hide{\oli{This lemma is indeed all we need for a Martingale or Lyapunov approach towards SSC. 
If a similar lemma is true for MW-$\alpha$, then in the case of {\bf iid} arrivals, we can extend our exponential SSC result (of Corollary \ref{prop:ssc iid trans}) to MW-$\alpha$ policies. 
Unfortunately, this is not the case, and the lemma doesn't hold for MW-$\alpha$ (there exist counterexamples). Added a short note to the discussion section.}}

\hide{It turns out that the perturbed trajectories are also attracted to the set of stable point, in a certain sense. This phenomenon is called SSC. In the transient regime, we will show in Subsection \ref{subsec:ssc trans} that a perturbed trajectory whose initial point is close to the set of critical points, stays close to this set for a long time and with a high probability. When then show in \ref{subsec:ssc std} that perturbed trajectories stay close the  this set also  in steady state. In particular, we show that the steady state distribution of distance from the set of invariant states decays exponentially with the distance.

Throughout this section, we assume that all arrival processes are stationary random processes. 
}

We continue with a definition that quantifies the 
rate at which a family of processes concentrates on its mean.

\begin{definition}[$f$-Tailed Sequence of Random Processes]\label{def:f tailed}
Consider a function $f:\N\times\R_+\to\R+$ and a vector $\lambda\in\R^n$. Let $A^r(\cdot)$ be a sequence of random processes indexed by $r\in\N$. Assume that for each $r$, $A^r(\cdot)$ is stationary, 
has expected value $\lambda^r$, 
and that $\lim_{r\to\infty}\lambda^r=\lambda$. Suppose that for every $\delta>0$, 
\begin{equation}\label{eq:def f-tailed arrival}
f(r,\delta)\, \prob{\frac{1}{r}\sup_{t\le r}  \Ltwo{\sum_{\tau=0}^t \big(A^r(\tau)- \lambda^r\big)} > \delta}\,  \xrightarrow[{\,r\to\infty\,}]{}\, 0.
\end{equation}
Then, $A^r(\cdot)$ is said to be an $f$-tailed sequence of random processes with limit mean $\lambda$, and we refer to $f$ as the \emph{concentration rate function}.
\end{definition}

\hide{Definition \ref{def:f tailed} resembles, for example, Assumption 2.5 in \cite{ShahW12}, when letting $f(r,\delta)=r\big(\log r\big)^2$, for all $\delta>0$.
However, the bound in Definition \ref{def:f tailed} is only on periods of length $r$, whereas in Assumption 2.5 of \cite{ShahW12} one needs to have  $f(z)\, \prob{\sup_{t\le z} \frac{1}{z} \Ltwo{\sum_{\tau=1}^t A^r(\tau)\,-\, t\lambda^r}\ge\delta_z}\xrightarrow[{\,z\to \infty\,}]{}0$ uniformly in $r$, which is a more restrictive condition. }

Later, we will show that 
the time scale over which SSC holds is almost  proportional to the best possible 
concentration rate function $f$.
We observe
 that any sequence of random processes that satisfies Assumption \ref{assumption 2.5} is a sequence of $f$-tailed processes,  with $f(r,\delta)=r\cdot \log^2 r$. 
However, the reverse is not true: Assumption~\ref{assumption 2.5} involves an additional requirement of uniform convergence over all values of an additional indexing parameter $z$, whereas  Definition \ref{def:f tailed} essentially only considers the case $z=r$. Thus, Definition \ref{def:f tailed} is less restrictive,  easier to check, and also seems  more natural.

\hide{
\comm{I am not sure about the discussion in the previous paragraph. The assumption of S+W is quite convoluted. Look at the following statements:\\
1. S+W covers the case of bounded iid processes\\
2. S+W implies a \pur{concentration} rate function $r\log^2 r$, that does not depend on $\delta$\\
3. Lemma 3 is ``tight'', that is, the dependence on $\delta$ cannot be removed\\
The above statements are contradictory. I do not think that $\delta$ can be eliminated from Lemma 3, so 3 is right. Also statemewnt 1 must be right, else S+W would be a useless paper. So, statement 2 must be false}

\oli{No contradiction. In fact, 3 together with Theorem 2 imply  1 and 2, that is the $r\log^2 r$ scaling for iid processes.
It is true that the term $\exp\big(\beta r\delta/a^2\big)$ in Lemma \ref{lem:iid tail} has a dependence on $\delta$. 
Nevertheless, $\exp\big(\beta r\delta/a^2\big)$ always dominates $r\log^2 r$ as $r$ goes to infinity,  for all values of $\delta>0$ and $\beta>0$.
Then, Lemma \ref{lem:iid tail} implies that iid processes are also $r\log^2 r$-tailed. 
Then, the S+W result for iid case, mentioned in 1 and 2 above, follows from our strong SSC result, Theorem~\ref{th:strong ssc}.}}

There are many processes whose concentration properties are well understood, and which translate to the requirements in Definition \ref{def:f tailed}, for a suitable concentration rate function $f$. We record one such fact in Lemma \ref{lem:iid tail} below, which deals with bounded i.i.d.\ arrival processes, and which is proved  in Appendix \ref{app:proof lem iid tail}. 

\begin{lemma}[Bounded I.I.D.~Processes are Exponential-Tailed] \label{lem:iid tail}
Fix a vector $\lambda\in \R_+^n$ and a constant $a>0$. Consider a sequence of random processes $A^r(\cdot)$ indexed by $r\in\N$.
Suppose that for every $r$, the random variables $A^r(t)$ are i.i.d., and that $A^r(t) \in [0,a]^n$, for all $t$. Denote the mean of $A^r(t)$ by $\lambda^r$, and suppose that $\lim_{r\to\infty} \lambda^r=\lambda$. Take any constant $\beta\in (0,2)$, and let $f\big(r,\delta\big)=\exp\big(\beta r\delta/na^2\big)$.  Then, $A^r(\cdot)$ is an $f$-tailed sequence of random processes with limit mean~$\lambda$.
\end{lemma}
Similar results are possible for arrival processes that are modulated by a finite and ergodic Markov chain. The boundedness assumption can also be removed under standard conditions on the moment generating function of $A^r(t)$.  

We now define processes involving a more general scaling of time, as a generalization of the fluid and diffusion-scaled processes.

\begin{definition}[$g$-Time-Scaled Processes] \label{def:g scaled solution}
Consider an increasing  function $g:\N\to\R_{+}$ and a sequence $Q^r(\cdot)$ of random  processes. Then, the corresponding  sequence of $g$-time-scaled processes $\qhat^r(\cdot)$ is defined as
\begin{equation}
\qhat^r(t) = \frac{1}{r} Q^r\Big(\big\lfloor g(r)t\big\rfloor\Big),
\end{equation}
for all $r\in\N$ and all $t\in\R_+$.
\end{definition}
\medskip
 
The fluid scaling and the diffusion scaling of a random process are particular $g$-time-scaled, processes corresponding to $g({r})={r}$ and $g({r})={r}^2$, respectively. 
Definition \ref{def:g scaled solution} 
allows for a more general
scaling of time. 

\subsection{Main SSC Result}
We now present our main SSC result. 

\begin{theorem}[Strong State Space Collapse]  \label{th:strong ssc}
Consider a network operating under a WMW policy, and a vector $\lambda$ in its capacity region, with a corresponding  set of invariant states $ \I(\lambda)$. Fix some $T\in\R_+$,
and let $\{\lambda^r\}$ be a sequence that converges to $\lambda$. 
Consider  two functions $f:\N\times\R_+\to\R_+$ and $g:\N\to\R_+$, with $\liminf_{r\to\infty}\, g(r)/r>0$. 
 Let $A^r(\cdot)$ be an $f$-tailed sequence of arrival processes with limit mean $\lambda$, and let $\qhat^r(\cdot)$ be a corresponding sequence of $g$-time-scaled queue length processes. 
Suppose that $d\Big(\qhat^r(0)\,,\, \I(\lambda)\Big) \to 0$,  as $r\to\infty$. 
\begin{enumerate}[label={(\alph*)}, ref={\ref{th:strong ssc}(\alph*)}]
\item \label{th:strong ssc 1}
Suppose that for every $\epsilon>0$,  we have $\liminf_{r\to\infty}\, rf\big(r,\epsilon\big)/g(r) > 0$. Then, for any $\delta>0$,
\begin{equation}\label{eq:sssc trans 1}
 \prob{\sup_{t\in  [0,T]}\,d\Big(\qhat^r(t)\,,\, \I(\lambda)\Big)>\delta}  \xrightarrow[{\,r\to\infty\,}]{} 0.
\end{equation}
\item \label{th:strong ssc 2}
Under the same assumptions as in Part (a), 
we can also bound the rate of convergence in \eqref{eq:sssc trans 1}: for any $\delta>0$, there exists an $\epsilon>0$ such that
\begin{equation}\label{eq:sssc trans 2}
\frac{rf\big(r,\epsilon\big)}{g(r)}\, \prob{\sup_{t\in  [0,T]}\,d\Big(\qhat^r(t)\,,\, \I(\lambda)\Big)>\delta}  \xrightarrow[{\,r\to\infty\,}]{} 0.
\end{equation}
Moreover, for the case of a MW policy, \eqref{eq:sssc trans 2} holds for every $\epsilon\le\min\big({\delta},\absorbconst\big)/2\divconst$, where $\divconst$ is the  sensitivity constant of the network (cf.~Theorem \ref{th:main sensitivity}) and $\absorbconst=\absorbconst(\lambda)$ is the constant in Lemma \ref{lem:absorb const}.
\item \label{th:strong ssc converse}
Conversely, suppose that $f:\N\times\R_+\to\R_+$ and $g:\N\to\R_+$, are such that $\lim_{r\to\infty}\, rf\big(r,\epsilon\big)/g(r) = 0$, for every $\epsilon>0$, 
and $\lim_{r\to\infty} g(r)/r =\infty$. 
Then, for any network operating under a MW policy, any arrival rate $\lambda$ in its capacity region (excluding its extreme points), and any $q_0\in\I (\lambda)$, there exists an $f$-tailed sequence of arrival processes satisfying \eqref{eq:def f-tailed arrival} 
and a corresponding sequence of $g$-time-scaled processes $\qhat^r(\cdot)$, $r\in\N$, initialized at $\qhat^r(0) = q_0$, such that
\begin{equation}\label{eq:sssc trans converse}
 \prob{\sup_{t\in  [0,T]}\,d\Big(\qhat^r(t)\,,\, \I (\lambda)\Big)>\delta}  \xrightarrow[{\,r\to\infty\,}]{} 1,
\end{equation}
for all $\delta>0$.
\end{enumerate}
\end{theorem}
The proof of Theorem \ref{th:strong ssc} is given in Section \ref{sec:proof ssc trans}. 
Part (b) relies on a reduction of WMW dynamics to MW dynamics together with  the facts that the queue length process stays close to a fluid solution (Theorem \ref{th:main sensitivity}), and that a fluid solution is attracted to the invariant set $\I$ (Lemma \ref{lem:absorb const}).
The proof of 
Part (c) relies on an explicit construction.

We note that Part (a) is a straightforward corollary of Part (b). Nevertheless, we have included the statement of Part (a) because it is in a form comparable to SSC results in the literature, and also because it facilitates a comparison with the converse result in Part (c).

\hide{\oli{Question: can we give a converse for Theorem \ref{th:strong ssc 1} in the following sense?: \\
``For every pair of functions $g(r)$ and $f(r,\epsilon)$ for which $\limsup_{r\to\infty}\, rf\big(r,\epsilon\big)/g(r) = 0$, there exists a network and a sequence of $f$-tailed arrival processes, for which the SSC \eqref{eq:sssc trans 1} does not happen for the $g$-scaling of the queue length process.''\\
We do this partially for a special case, in Example \ref{ex:ssc iid}, where $f$ is exponential (corresponding to iid arrivals) and $g$ is any super-exponential function.\\
Is it worth thinking about this question?}}

Theorem \ref{th:strong ssc} ties together the time scaling $g$ over which SSC occurs and the concentration rate function, $f$, of the arrival processes.
The underlying intuition is that if the queue length process is initialized sufficiently close to $\I$, then it will stay in an  $r\delta$-neighborhood of $\I$, with high probability, for a period of time proportional to $g(r)$.
This  enables us to prove additive SSC over time scales much longer than  those underlying the diffusion scaling, as in the next subsection.

\subsection{Special Cases of SSC}
In this section, we apply Theorem \ref{th:strong ssc} to obtain more concrete SSC results. 
The first result concerns SSC over an exponentially large time scale. While it refers to bounded i.i.d.\ processes, it admits straightforward extensions to 
arrival processes with a concentration rate function $f$ that
 grows exponentially with $r$, as is the case whenever a suitable Large Deviations Principle holds. 

\begin{corollary}[Bounded I.I.D.\ Arrivals: SSC over an Exponential Time Scale] \label{prop:ssc iid trans}
Consider a network operating under a MW policy, 
 a vector $\lambda$ in its capacity region, a $\delta>0$, and a sequence $A^r(\cdot)$ of arrival processes that  satisfy the assumptions of Lemma \ref{lem:iid tail}. 
Consider a $\gamma< \min\big(\delta,\absorbconst\big)/\big(2\divconst n a^2\big)$, where $\divconst$ is the input sensitivity constant of the network, $\absorbconst=\absorbconst(\lambda)$ is the constant in Lemma \ref{lem:absorb const}, and $a$ is an upper bound on the size of arriving jobs (cf. Lemma \ref{lem:iid tail}). 
Consider the $e^{\gamma r}$-time-scaling of the queue length processes,
\begin{equation}\label{eq:exp scale of time in cor1}
\qhat^r(t) = \frac{1}{r} Q^r\Big(\big\lfloor e^{\gamma r}t\big\rfloor\Big),
\end{equation}
and suppose that $d\Big(\qhat^r(0)\,,\, \I(\lambda)\Big) \to 0$,  as $r\to\infty$. 
Then, for any $T\in\R_+$,
\begin{equation}\label{eq:ssc iid}
e^{\gamma r} \,\prob{\sup_{t\in  [0,T]}\,d\Big(\qhat^r(t)\,,\, \I(\lambda) \Big)>\delta}  \xrightarrow[{\,r\to\infty\,}]{} 0.\ \ 
\end{equation}
\end{corollary}

\begin{proof} Let
$\beta=\gamma /\big[ \min\big(\delta,\absorbconst\big)/\big(2\divconst n a^2\big)\big]$. Then, $\beta< 1$, and
 Lemma \ref{lem:iid tail} implies that  $A^r(\cdot)$ is an $f$-tailed sequence of processes for $f\big(r,\epsilon\big) = \exp\big( 2\beta r \epsilon / na^2\big)$. 
Let $\epsilon = \min\big(\delta,\absorbconst\big)/2\divconst $. 
Then, $\gamma = \beta \epsilon/n a^2$. 
Let $g(r) = \exp\big(\gamma r\big) = \exp\big(\beta r \epsilon/na^2\big)$ be the time scaling in the definition \eqref{eq:exp scale of time in cor1} of $\qhat^r(t)$. 
Then,
\begin{equation}
 \frac{rf\big(r,\epsilon\big)}{g(r)} \, = \, \frac{r\exp\big(2\beta r\epsilon /n a^2\big)}{\exp\big(\beta r\epsilon /n a^2\big)} \,=\, r \exp\big( \beta r \epsilon / n a^2\big) \,> \,\exp(\gamma r).
\end{equation} 
Therefore, the assumptions in Part (b) of Theorem \ref{th:strong ssc 2} are satisfied, and
\hide{
\begin{equation} \label{eq:limsup in cor1}
\limsup_{r\to\infty}\, \frac{rf\big(r,\epsilon\big)}{g(r)} \,\prob{\sup_{t\in  [0,T]}\,d\Big(\qhat^r(t)\,,\, \I(\lambda) \Big)>\delta}
  \,=\, 0.
\end{equation}
where $g(r) = \exp\big(\gamma r\big)$ is the scale of time in the definition of $\qhat^r(t)$, in \eqref{eq:exp scale of time in cor1}.
Moreover, for the above choices of $\epsilon$ and $\beta$, we have  $\gamma = \beta \epsilon/a^2$.
Then,
Therefore,
}
\begin{equation}
\begin{split}
&\limsup_{r\to\infty}\,e^{\gamma r} \,\prob{\sup_{t\in  [0,T]}\,d\Big(\qhat^r(t)\,,\, \I(\lambda) \Big)>\delta} \\
&\qquad \le\, 
\limsup_{r\to\infty}\, \frac{rf\big(r,\epsilon\big)}{g(r)} \,\prob{\sup_{t\in  [0,T]}\,d\Big(\qhat^r(t)\,,\, \I(\lambda) \Big)>\delta}
  \,=\, 0.
  \end{split}
\end{equation}
 Thus, 
\eqref{eq:ssc iid} holds, which is the desired result.
\end{proof}

We note that Part (c) of Theorem~\ref{th:strong ssc} provides a partial converse to Corollary~\ref{prop:ssc iid trans}:
under i.i.d.\ arrivals with nonzero variance,  additive SSC  
does not  hold over a 
super-exponential time scale. 

The next corollary of Theorem \ref{th:strong ssc}{(a)} concerns additive SSC under  diffusion scaling.
\begin{corollary}[State Space Collapse in Diffusion Scaling] \label{prop:ssc diffusion}
Consider a network operating under a WMW policy, and a function 
 $f:\N\times\R_+\to\R_+$  such that $ \liminf_{r\to\infty}\, f(r,\delta)/r>0$, for all $\delta>0$. Consider a $\lambda$ in the capacity region, an  $f$-tailed sequence of arrivals $A^r(\cdot)$ with limit mean $\lambda$, and a corresponding diffusion-scaled queue length processes $\qhat^r(\cdot)$ (cf. \eqref{eq:def diffusion scale}).
Suppose that $d\Big(\qhat^r(0)\,,\, \I(\lambda)\Big) \to 0$,  as $r\to\infty$. 
Then, for any $T\in\R_+$,
\begin{equation}
\prob{\sup_{t\in  [0,T]}\,d\Big(\qhat^r(t)\,,\, \I (\lambda)\Big)>\delta}  \xrightarrow[{\,r\to\infty\,}]{} 0.
\end{equation}
\end{corollary}
 \medskip
 

Corollary  \ref{prop:ssc diffusion} strengthens 
Theorem \ref{conj:shah 7.2}, for the case of WMW policies, in that the assumption of  i.i.d.~arrivals is  removed. We only require a concentration property for the arrival process, such as
\begin{equation}
r\, \prob{\sup_{t\le r} \frac{1}{r} \Ltwo{\sum_{\tau=0}^t\big( A^r(\tau)\,-\, \lambda\big)}\ge\delta}\xrightarrow[{\,r\to \infty\,}]{}0,  \qquad \forall \delta>0,
\end{equation}
which is  even weaker than Assumption~\ref{assumption 2.5}.
Moreover, under a MW policy and i.i.d.~arrivals, 
our Corollary \ref{prop:ssc iid trans} extends 
Theorem \ref{conj:shah 7.2} by establishing SSC over an exponential time scale (as opposed to the diffusion scaling). For further perspective with respect to existing results, please refer to the discussion following the statement of Theorem~\ref{conj:shah 7.2}, in Section~\ref{s:ssc-lit}.


\hide{
\oli{In the same vein, Theorem 7.1 of \cite{ShahW12} can be strengthened, to give additive, instead of multiplicative, state space collapse.}
\comm{I already put some comparisons right after the statement of the conjecture. Let us keep all comments there. I would suggest no more comments here other than "for perspective with respect to existing results, please refer to the discussion following the stament of the conjecture, in Section ...}

\hide{
\begin{corollary}
Conjecture \ref{conj:shah 7.2} is  true for the case of WMW policies.
\end{corollary}
\begin{proof}

\end{proof}
}

\oli{
\begin{itemize}
\item[1-]
Equation (\ref{eq:def f-tailed arrival}) only bounds $A^r(\cdot)$ for periods of length $r$, whereas in Assumption 2.5 of \cite{ShahW12} one needs to have  ${f'}(r)\, \prob{\sup_{t\le r} \frac{1}{r} \Ltwo{\big(\sum_{\tau=1}^t A^z(\tau)\,-\, \lambda^z\big)}\ge\delta_r}\xrightarrow[{\,r\to \infty\,}]{}0$ uniformly in $z$, for $f'(r)=r(\log r)^2$. \comm{I don't understand. The S+W version also bounds for just a period of length $r$.}

\item[2-] 
In Corollary \ref{prop:ssc diffusion}, $f(r)$ can be any linear or supper-linear function of $r$, whereas in Assumption \ref{assumption 2.5}, $f'(r)\ge r(\log r)^2$, which is a more restrictive assumption.

\item[3-] Theorem \ref{th:strong ssc} assumes ${d\big(\qhat^r(0)\,,\, \I\big)} \to 0$ which is more general compared to the assumption $\qhat^r(0)\to q_0\in\I$ as $r\to \infty$ in Conjecture \ref{conj:shah 7.2}.
\end{itemize}
}
}

\section{\bf Proof of Theorem \ref{th:main sensitivity}}\label{sec:proof sensitivity}
In this section, we present the proof of Theorem \ref{th:main sensitivity}, organized in a sequence of  subsections.  
We first show in Subsection \ref{subsec:WMW-reduction} that for any network operating under a WMW policy, there is another network operating under a MW policy whose queue length  process is a linear transformation of the queue length process of the original network. 
Thus, we can just focus on the MW policy.
In Section \ref{subsec:fpcs} we review a general sensitivity result on a class of dynamical systems with piecewise constant drift. Next, in Subsection \ref{subsec:net to fpcs} we introduce an \emph{induced} continuous-time dynamical system that provides the bridge between the original discrete-time process under a MW policy and the fluid model. 
The proof concludes in Subsection \ref{subsec:proof sensitivity th} by applying the general sensitivity result to the induced system.


\subsection{From WMW to MW} \label{subsec:WMW-reduction}
In order to  leverage the tools that we will develop for  MW policies  and apply them to the more general WMW policies, we start with a reduction from WMW policies to a MW policy. 
This is accomplished through 
the following lemma, which shows that the queue lengths and fluid solutions under a WMW policy  are linear transformations of queue lengths and fluid solutions under a MW policy, in a transformed network.

\begin{lemma}[Reduction of WMW Dynamics to MW Dynamics]\label{lem:wmw2mw}
Consider a network $\net$ with action set $\S$ and a routing matrix $\RM$. Fix a weight vector $w$, an arrival function $A(\cdot)$, and an arrival rate vector $\lambda$. Let $Q(\cdot)$  be a queue length process  of $\net$ corresponding to the arrival $A(\cdot)$,  
under a $w$-WMW policy. 
Let $W=\diag(w)$, $\tilde{\lambda}=W^{1/2} \lambda$, 
and $\tilde{A}(t)=W^{1/2} A(t)$, for all $t\in\Z_+$. Let $\tilde{\net}$ be  a network with  action set $\tilde{\S}=W^{1/2} \S$ and  routing matrix $\tilde{\RM} = W^{1/2} \RM W^{-1/2}$.
Then, 
\begin{enumerate}[label={(\alph*)}, ref={\ref{lem:cp}(\alph*)}]
\item \label{lem:wmw2mw Q}
$\tilde{Q}(t)=W^{1/2} Q(t)$ is a queue length process of $\tilde{\net}$ corresponding to the arrival $\tilde{A}(\cdot)$, 
under a MW policy.
\item \label{lem:wmw2mw q}
$\tilde{q}(t)=W^{1/2} q(t)$ is a fluid solution of  $\tilde{\net}$ corresponding to   arrival rate $\tilde\lambda$ and unit weights (as in MW)
if and only if $q(t)$ is a  fluid solution of  $\net$ corresponding to  arrival rate $\lambda$ and WMW weights $w$.

\end{enumerate}
\end{lemma}
\begin{proof}
Given some $\mu\in \S$ and $Q\in\R_+^n$, 
we let 
$\tilde{\mu}=W^{1/2}\mu$ and $\tilde{Q} = W^{1/2}Q$. Then, 
\begin{eqnarray*}
\tilde{Q}^T\big(I-\tilde\RM\big) \tilde{\mu}  &=&  \big(W^{1/2} Q\big)^T  \big(I-W^{1/2}\RM W^{-1/2}\big) \,W^{1/2}\mu       \\
&=&  Q^T \,W^{1/2} W^{1/2}  \big(I-\RM \big) \,W^{-1/2} W^{1/2}\,\mu \\
&=& Q^T W \big(I-\RM\big) \mu.
\end{eqnarray*}
Therefore, $\tilde\mu\in \tilde{\S}$ is a maximizer of $\tilde{Q} \big(I-\tilde\RM\big) \tilde{\mu} $ if and only if $\mu\in\S$ is a maximizer of $Q^T W \big(I-\RM\big) \mu$, i.e., $\tilde{\S}(\tilde{Q}) = W^{1/2} \S_w(Q)$.

For Part (a), for any $t\in\Z_+$,
\begin{eqnarray*}
\tilde{Q}(t+1) &=& W^{1/2} Q(t+1)\\
&=& W^{1/2} Q(t)   \,+\, W^{1/2} A(t) \,+\, W^{1/2} \big(\RM-I\big) \min\big(  \mu(t),  Q(t)\big)\\
&=& \tilde{Q}(t)   \,+\, \tilde{A}(t) \,+\, W^{1/2} \big(\RM-I\big) W^{-1/2} \, W^{1/2}\min\big( \mu(t),Q(t)\big)\\
&=& \tilde{Q}(t)   \,+\, \tilde{A}(t) \,+\, \big(\tilde{\RM}-I\big) \min\big( \tilde{\mu}(t),\tilde{Q}(t)\big).
\end{eqnarray*}
Therefore, $\tilde{Q}(\cdot)$ satisfies the evolution rule \eqref{eq:network evolution rule} of $\tilde{\net}$, and is a queue length process corresponding to the arrival function $\tilde{A}(\cdot)$. Since $Q(t)$ evolves according to a $w$-WMW policy, we have $\mu(t)\in \S_w\big(Q(t)\big)$. As shown earlier, this implies that $\tilde{\mu}(t)\in\tilde{\S}(\tilde{Q})$, and thus $\tilde Q (t)$ indeed follows a MW policy.

For Part (b), consider a set of functions $y(\cdot)$ and $s_\mu(\cdot)$ for $\mu\in \S$, that together with $q(\cdot)$ satisfy \eqref{eq:fluid diff eq 1}--\eqref{eq:fluid diff eq 6}.
It is not difficult to see that all equations remain valid when $q(\cdot)$, $\lambda$, $\S$, $s_\mu(\cdot)$,  $y(\cdot)$, and $w$  are replaced with $\tilde{q}(\cdot)$, $\tilde{\lambda}$, $\tilde{\S}$, $s_{\tilde\mu}(\cdot)=s_{\mu}(\cdot)$, $W^{1/2}y(\cdot)$, and $\ones_n$, respectively. The reverse  direction is also true.
Therefore, $\tilde{q}(\cdot)$ is a fluid solution of $\tilde{\net}$ corresponding to the arrival rate vector  $\tilde{\lambda}$, with unit weights, if and only if ${q}(\cdot)$ is a fluid solution of ${\net}$ corresponding to the arrival rate vector  $\lambda$, with weight vector $w$.
\end{proof}


\subsection{FPCS Dynamical Systems}\label{subsec:fpcs}
In this subsection, we review some definitions and  results from \cite{AlTG19sensitivity}. 
A dynamical system is identified with a set-valued function $F:\R^n\to 2^{\R^n}$
and the associated differential inclusion $\dot{x}(t)\in F(x(t))$. We start with a formal definition, which allows for the presence of perturbations.

\begin{definition}[Trajectories of a Dynamical System]\label{def:integral pert traj} 
Consider a dynamical system $F:\R^n\to 2^{\R^n}$, and let $\Prt:\R\to\R^n$ be a right-continuous function, which we refer to as the \emph{perturbation}.
 Suppose that  $\ptraj({\cdot})$ and $\ff({\cdot})$ 
are measurable functions 
 of time that satisfy \begin{equation}\label{eq:def of integral pert 1}
\ptraj(t) = \int_0^t \ff(\tau)\,d\tau \,+\, \Prt(t),\qquad \forall\ t\ge 0,
\end{equation}
\begin{equation}\label{eq:def of integral pert 2}
\ff(t)\in F\big(\ptraj(t)\big),\qquad \forall\ t\ge 0.
\end{equation}
We then call  $\ptraj$ a 
\emph{perturbed trajectory}
 corresponding to $\Prt$.
  In the special case where 
$U$ is  identically zero, we also refer to $\ptraj$  as an \emph{unperturbed trajectory}.
\end{definition}

For a convex function $\Phi:\R^n\to\R$, we denote its subdifferential by $\partial \Phi(x)$. 
We say that $F$ is a \emph{subgradient dynamical system} if there exists a convex function $\Phi:\R^n\to\R$, 
such that for any $x\in\R^n$, $F(x)=-\partial\Phi(x)$.
Furthermore, if $\Phi$ is of the form
$$\Phi(x)=\max_{i}\big(-\mu_i^Tx+b_i\big),$$ 
for some $\mu_i\in\R^n$, $b_i\in\R$, and with $i$ ranging over a {\bf finite} set,
we say that $F$ is a \emph{Finitely Piecewise Constant  Subgradient} (FPCS, for short) system. Note that for such systems, $F(x)$ is always equal to the convex hull of the vectors $\mu_i$ that maximize $-\mu_i^T x+ b_i$.

FPCS systems admit a very special sensitivity bound.
\begin{theorem}[\cite{AlTG19sensitivity} Theorem 1] \label{th:sensitivity fpcs}
Consider an FPCS   system $F$.
Then, there exists a constant $\divconst$ such that for any unperturbed trajectory $x(\cdot)$, and for any 
perturbed trajectory $\ptraj( \cdot)$ with corresponding  perturbation $\Prt( \cdot)$ and the same initial conditions $\ptraj(0)=x(0)$,  we have
\begin{equation} \label{eq:bounded pert cont}
\Ltwo{\ptraj(t)- x(t)} \,\le \, 
\divconst\, \sup_{\tau\le t} \Ltwo{ \Prt(\tau)}, \qquad \forall\ t\in \R_+.
\end{equation}
Moreover, for any $\lambda\in\R^n$, the bound \eqref{eq:bounded pert cont} applies to the 
(necessarily FPCS)   system 
$F(\cdot)+\lambda$ 
 with the same constant  $\divconst$.
\end{theorem}

\subsection{Reduction of the MW Dynamics to an FPCS System}\label{subsec:net to fpcs}
Throughout this subsection, we restrict attention to a network operated under an (unweighted) MW policy.
In order to take advantage of Theorem \ref{th:sensitivity fpcs}, we  show that a  discrete-time network can also be represented as an associated (``induced'') FPCS dynamical system. 

\begin{definition}[Induced FPCS system] \label{def:induced net}
For a network with action set $\S$ and routing matrix $\RM$, the induced FPCS system is the subgradient dynamical system $F$ associated with the convex function
\begin{equation} \label{eq:def phi for fpcs}
\Phi(x)=\max_{\mu\in \S}\, \big((I-\RM)\mu\big)^Tx.
\end{equation}
In particular, $F(x)$ is the convex hull of the image of $\S(x)$ under the linear transformation $R-I$, where $\S(x)$ is the set of vectors $\mu\in\S$ that maximize $\big((I-R)\mu\big)^Tx$.
\end{definition}


We start with the observation that fluid solutions of a network are  trajectories of the induced FPCS system. 
Roughly speaking, this is because the service vectors chosen by the MW policy in \eqref{eq:maximal} are maximizers of the set of  linear functions $\big((I-\RM)\mu\big)^T Q $ over $\mu\in\S$, and the fluid solution moves along the negative of a convex combination of such maximizing service vectors. Thus,  fluid solutions move along the subgradients of  $\Phi$ (defined in \eqref{eq:def phi for fpcs}), and are therefore  trajectories of the induced FPCS system.

\begin{proposition}[Fluid Model Solutions as Trajectories  of the Induced FPCS System] \label{prop:simulation fluid}
Consider a network and its induced FPCS system $F$. Let $q(\cdot)$ be a fluid solution of the network corresponding to arrival rate $\lambda$.
Then, $q(\cdot)$ is an unperturbed trajectory of the dynamical system $\dot{q} \in F(q)+\lambda$.
Conversely, any  unperturbed trajectory $x(\cdot)$ of $F(\cdot)+\lambda$, with $x(0)\in\R_+^n$, is a fluid solution corresponding to $\lambda$.
\end{proposition}
\begin{proof}
 For a vector $\mu\in\R_+^n$ and a set $J\subseteq\big\{ 1,\ldots,n\big\}$ of indices, we let 
\begin{equation} \label{eq:def DJ mu}
D_J(\mu)\triangleq \left\{ \xi\in\R_+^n \,\Big|\, \xi_i=\mu_i, \textrm{ for all } i\not\in J, \textrm{ and } 0\le\xi_j\le \mu_j, \textrm{ for all }j\in J   \right\}.
\end{equation}
Equivalently,
\begin{equation}\label{eq:Dj as convex hull}
D_J(\mu) \,  =\, \conv\Big({\big\{ \sigma_{-K}(\mu) \,\big|\,  K\subseteq J  \big\}  }\Big),
\end{equation}
where $\sigma_{-K}(\mu)$ is a vector whose $i$th entry is equal to the $i$th entry of $\mu$ if $i\not\in K$, and equal to zero if $i\in K$.
Recall that $\S(q)$ is defined as the set of all $\mu\in\S$ that maximize $\big((I-R)\mu\big)^T q$; cf.\
\eqref{eq:maximal}.

\begin{claim}\label{claim:fluid}
Fix a $q\in\R_+^n$ and a $\mu\in \S(q)$. Let $J=\big\{j \,\big|\, q_j=0\big\}$. Then, 
\begin{equation} \label{eq:claim statement D subset F}
\big(\RM-I\big) D_J(\mu)\subseteq F(q).
\end{equation}
\end{claim}
\begin{proof}[Proof of Claim]
\hide{We are assuming that $\mu\in \S(q)$, i.e., $\mu$ is a maximizer of the objective  $\big((I-R)\mu\big)^T q$, over all $\mu\in\S$. 
For any $j$ such that $q_j=0$, the contribution of the term $\mu_j$ to this objective
 comes only through the term $\big((-R)\mu\big)^T q$.
 Since $R$ and $q$ are both non-negative, if we reduce the value of $\mu_j$ to a smaller non-negative value,
 the objective can only improve. By repeating this argument  for all components $\mu_j$ with $j\in J$, we see that any vector $\mu\in D_J(\mu)$ is also a maximizer
 }
Note that for any $\mu\in\S$ and any set $K$ of indices, the vector $\sigma_{-K}(\mu)$ also belongs to $\S$, because of Assumption \ref{assumption 1}. We now fix some $q$ and the set $J$, as in the statement of the claim. For any $K\subseteq J$, we have
$\sigma_{-K}(\mu) \preceq \mu$. Furthermore, since  the entries of $\RM$ and $\mu$  are non-negative, we have  $q^T \RM \,\sigma_{-K}(\mu) \le q^T \RM\, \mu$. Therefore, 
\begin{equation*}
\begin{split}
q^T \big(I- \RM\big) \sigma_{-K}(\mu) \,& =\,  q^T \sigma_{-K}(\mu) -q^T \RM\, \sigma_{-K}(\mu)\\
&=\,  q^T \mu -q^T \RM \,\sigma_{-K}(\mu)\\
&\ge\,  q^T \mu -q^T \RM \, \mu\\
&=\, q^T \big(I-\RM\big) \mu,
\end{split}
\end{equation*}
where the second equality holds because $q_j=0$  whenever the $j$th entry of $\sigma_{-K}(\mu)$ is not equal to $\mu_j$.  
We have therefore established that if $ \mu\in\S(q)$, then $\sigma_{-K}(\mu)\in \S(q)$. 
Since $F(q)$ is the image under $\RM-I$ of the convex hull of $\S(q)$, we obtain 
\begin{equation} \label{eq:I-R sigma k mu in f}
\big(\RM-I\big) \sigma_{-K}(\mu) \in F(q).
\end{equation}
Therefore, 
\begin{eqnarray*}
\big(\RM-I\big) D_J(\mu)
&=&  \big(\RM-I\big)\, \conv\Big({\big\{ \sigma_{-K}(\mu) \,\big|\,  K\subseteq J  \big\}  }\Big)\\
&=& \conv\Big({\big\{ \big(\RM-I\big) \sigma_{-K}(\mu) \,\big|\,  K\subseteq J  \big\}  }\Big) \\
&\subseteq& \, F(q),
\end{eqnarray*}
where the first equality is due to  \eqref{eq:Dj as convex hull}, and the last relation is due to \eqref{eq:I-R sigma k mu in f} and the convexity of $F(q)$. This establishes the validity of the claim
\eqref{eq:claim statement D subset F}.
\end{proof}
\medskip

We now return to the proof of the proposition.
Consider a function $y(\cdot)$ and a set of functions $s_\mu(\cdot)$, $\mu\in \S$, that together with $q(\cdot)$ satisfy the fluid model equations \eqref{eq:fluid diff eq 1}--\eqref{eq:fluid diff eq 6}. 
We will show that  $q(\cdot)$ satisfies the differential inclusion $\dot{q}\in F(q)+\lambda$.
Fix some $t\ge 0$ and let $y=y(t)$.
For any $\mu\in\S$, let $s_\mu=s_\mu(t)$ and consider an $n$-dimensional vector $y^\mu$ with entries 
\begin{equation*}
y^\mu_i \,=\, \begin{cases} {y_i\mu_i }/\,{ \sum_{\nu\in \S}s_\nu \nu_i},\, & \textrm{if }\sum_{\nu\in \S}s_\nu \nu_i\ne 0,\\ 0,& \textrm{otherwise}.\end{cases}
\end{equation*}
It follows that for $i=1,\ldots,n$,
\begin{equation}
\sum_{\mu\in \S}s_\mu y^\mu_i \,=\, \sum_{\mu\in \S}s_\mu \frac{y_i\mu_i}{\sum_{\nu\in \S}s_\nu \nu_i} \, =\, y_i.
\end{equation}
Then,
\begin{equation} \label{eq:y as a convex commbination of y mu}
\sum_{\mu\in \S}s_\mu y^\mu \,=\, y.
\end{equation}
On the other hand, for any $\mu\in\S$ and for $i=1,\ldots,n$, either $y_i^{\mu}=0$ or \eqref{eq:fluid diff eq 4} implies that $y^\mu_i = \mu_i \big(y_i / \sum_{\nu\in \S}s_\nu \nu_i\big)\le \mu_i$. 
Moreover, for any $\mu\in\S$ and any $i\le n$, if  $q_i(t)>0$, then from \eqref{eq:fluid diff eq 5}, $y^\mu_i = y_i \big(\mu_i / \sum_{\nu\in \S}s_\nu \nu_i\big) =0$. 
Letting $J=\big\{j \,\big|\, q_j(t)=0\big\}$, it then follows from the definition of $D_J(\mu)$ that for any $\mu\in\S$, $\mu- y^\mu\in  D_J(\mu)$.
Claim \ref{claim:fluid} then implies that $\big(\RM-I\big) (\mu-y^\mu) \in F\big(q(t)\big)$.  
Therefore, in light of \eqref{eq:fluid diff eq 3} and the convexity of $ F\big(q(t)\big)$, we have
\begin{equation}
\sum_{\mu\in \S}s_\mu \big(\RM-I\big) (\mu-y^\mu) \, \in\,  F\big(q(t)\big).
\end{equation}
Finally, from \eqref{eq:fluid diff eq 1},
\begin{equation}
\begin{split}
\dot{q}(t)  
&\,= \, \lambda + \big(\RM-I\big) \left(\sum_{\mu\in\S}s_{\mu}\mu - {y}\right)\\
&\,= \, \lambda + \big(\RM-I\big) \left(\sum_{\mu\in\S}s_{\mu}\mu - \sum_{\mu\in \S}s_\mu y^\mu\right)\\
&\, =\, \lambda \,+\, \sum_{\mu\in \S} s_\mu \big(\RM-I\big) (\mu-y^\mu)\\
&\,\in\, F\big(q(t)\big) \,+\, \lambda,
\end{split}
\end{equation}
where the second equality is due to \eqref{eq:y as a convex commbination of y mu}.

\hide{
\oli{****DELETE FROM HERE*******}

Fix some $t\ge 0$ and let $\theta = \sum_{\mu\in \S} s_\mu(t) \mu - y(t)$.
Then,  \eqref{eq:fluid diff eq 4} implies that $\theta\succeq 0$.
Let  $J=\big\{j \,\big|\, q_j(t)=0\big\}$. 
It follows from \eqref{eq:fluid diff eq 5} that $\theta \in D_J(\mu)$.
\red{[This argument needs work, has to be rewritten. The issue is that $\theta$ is  defined through a sum over all $\mu$, hence you cannot calim that it has a property $\theta \in D_J(\mu)$ for a specific $\mu$ (which particular $\mu?)$]}
Claim \ref{claim:fluid} then implies that $(I-\RM) \theta \in F\big(q(t)\big)$.  
Finally, from \eqref{eq:fluid diff eq 1} we have
\begin{equation}
\dot{q}(t)  
\, =\, \lambda \,+\,\big(\RM-I\big) \theta 
\,\in\, F\big(q(t)\big) \,+\, \lambda.
\end{equation}
Therefore, $q(\cdot)$ is an unperturbed trajectory of $F(\cdot)+\lambda$.

\oli{****UP TO HERE*******}
}

\hide{****DELETE FROM HERE*******

Fix some $t\ge 0$ and let $J=\big\{j \,\big|\, q_j(t)=0\big\}$. 
 vector $\mu\in\S\big(q(t)\big)$, we define a vector $\tilde\mu\in\R^n$ with entries
\begin{equation}\label{eq:def tilde mu}
\tilde{\mu}_i = \left(1-\frac{y_i(t)}{\sum_{\pi\in\S}s_\pi(t)\mu_i}  \right), \quad i=1,\ldots,n.
\end{equation}
\comm{I cannot follow because the same symbol $\mu$ is used for the fixed vector of interest, but also as a free variable. The way it is written, the denominator
$\sum_{\pi\in\S}s_\pi(t)\mu_i$ should be equal to just $\mu_i$. Also, in the next equation, we should have a term $\mu_i\cdot\tilde\mu_i$. So, something is wrong.} \oli{Please see the purple fix above.}
For a $\mu\in \S$ if  $s_\mu(t)>0$, then \eqref{eq:fluid diff eq 6} implies that $\mu\in \S\big(q(t)\big)$. Therefore, for any $i\le n$,
\begin{equation}\label{eq:sum s - y = sum mu tilde}
\sum_{\mu\in \S} s_\mu(t) \mu_i - y_i(t) \,=\,  \sum_{\mu\in \S(q(t))} s_\mu(t) \mu_i -y_i(t) \,=\,  \sum_{\mu\in \S(q(t))} s_\mu(t) \tilde\mu_i.
\end{equation}

On the other hand, it follows from \eqref{eq:fluid diff eq 4} that $\tilde\mu \in \R_+^n$, and from  \eqref{eq:fluid diff eq 5} that $\tilde{\mu}\in D_J(\mu)$. Claim \ref{claim:fluid} then implies that $\big(\RM-I\big) \tilde{\mu}\in F\big(q(t)\big)$.  
Moreover, \eqref{eq:fluid diff eq 3} and \eqref{eq:fluid diff eq 6} certify that $\sum_{\mu\in \S(q(t))} s_\mu=1$. 
Therefore, it follows from the convexity of $F\big(q(t)\big)$ that
\begin{equation} \label{eq:sum mu in F}
(I-\RM) \sum_{\mu\in \S(q(t))} s_\mu(t) \tilde\mu
\,\in\,   \conv\Big({\Big\{ \big(I-\RM  \big) \tilde{\mu} \,\big|\,  \mu\in\S\big(q(t)\big) \Big\}  }\Big) 
\,\subseteq\, F\big(q(t)\big),
\end{equation}
\comm{Rewrite the above. Avoid the definition \eqref{eq:def tilde mu} (which could also involve a 0/0 issue). Use notation $\theta(\mu)$ instead of $\tilde \mu$, to make it clear that $\tilde \mu$ depends on $\mu$ and so that the meaning of the various summations is clear. Just define $\theta(\mu)$ by requiring \eqref{eq:sum s - y = sum mu tilde} to hold.
Then say that such $\theta(\mu)$ are guaranteed to exist, are non-negative, are smaller than the $\mu$, etc.} \oli{Please see the purple fix above.}

Finally, from \eqref{eq:fluid diff eq 1} we have
\begin{eqnarray*}
\dot{q}(t) & = & \lambda + \big(\RM-I\big) \left(\sum_{\mu\in\S}s_{\mu}(t)\mu - {y}(t)\right)\\
& =&  \lambda \,+\,\big(\RM-I\big) \sum_{\mu\in \S(q(t))} s_\mu(t) \tilde\mu \\
& \in & F\big(q(t)\big) \,+\, \lambda ,
\end{eqnarray*}
where the second  equality is due to \eqref{eq:sum s - y = sum mu tilde}.
Therefore, $q(\cdot)$ is an unperturbed trajectory of $F(\cdot)+\lambda$.

\pur{*****DELETE UP TO HERE******}}

We now prove the converse
 part of the proposition, that every unperturbed  trajectory $x(\cdot)$ of $F(\cdot)+\lambda$, initialized in the positive orthant, is a fluid solution. 
Consider  a fluid solution $q(\cdot)$ corresponding to the arrival rate $\lambda$, initialized with $q(0)=x(0)$ 
(for proofs of existence, see Appendix~A of \cite{DaiL05} and Lemma~9 of \cite{MarkMT18}).

It then follows from the first part of this proposition that $q(\cdot) $ is also an unperturbed trajectory of $F(\cdot)+\lambda$. On the other hand, 
it is shown in \cite{Rock70} that any subgradient dynamical system is a maximal monotone map\footnote{A set-valued function $F:\R^n\to 2^{\R^n}$ is a monotone map if for any $x_1,x_2\in\R^n$ and any $v_1\in F(x_1)$ and $v_2\in F(x_2)$, we have $\big(v_1-v_2\big)^T\big(x_1-x_2\big)\le 0$. 
It is called a maximal monotone map if it is monotone, and for any monotone map $\tilde{F}$, that satisfies $F(x)\subseteq \tilde{F}(x)$ for all $x$, we have $\tilde{F}=F$.}. 
Then,  Corollary 4.6 of \cite{Stew11} implies that there is a unique unperturbed trajectory with initial point $x(0)$. 
Therefore, $x(t)=q(t)$, for all $t\ge 0$, and the desired result follows. 
\end{proof}

Proposition \ref{prop:simulation fluid} has established that  a fluid solution is a trajectory of the  induced FPCS system. We now show, in the next proposition, 
that the actual discrete-time queue length process is close to  
a perturbed trajectory of the induced FPCS system.
Note that even if the discrete-time system has completely deterministic and steady arrivals (no stochastic fluctuations) it can still ``chatter'' around the boundary separating two regions with different drifts. The idea behind the proof is that this chattering can also be viewed as a perturbation of a straight trajectory.
This is conceptually straightforward, but some of the details 
of the behavior in the vicinity of such boundaries are tedious.

\begin{proposition}[Queue Length Processes as Trajectories  of the Induced FPCS System] \label{prop:simulation}
For any network, there exists a constant $\beta$ that satisfies the following statement.
Fix a $\lambda\in\R_+^n$, and let $A(\cdot)$ be an arrival function and $Q(\cdot)$  be a corresponding queue length process. Then, there exists a right-continuous (perturbation) function $\Prt(\cdot):\R_+\to\R^n$, satisfying 
\begin{equation} \label{eq: U bounded cond}
\sup_{\tau\le t} \Ltwo{\Prt(\tau)}\, \le\, \lVert  \lambda \rVert + \beta 
+ \sup_{\tau\le t} \Ltwo{ \sum_{\substack{k < \tau \\ k\in\Z_+ }} \big(A(k) -\lambda   \big) },\qquad \forall\ t\in \R_+,
\end{equation}
and a corresponding perturbed trajectory $\ptraj(\cdot)$ of $F(\cdot)+\lambda$ such that 
\begin{equation} \label{eq: X=Q cond}
\big\| \ptraj(k) - Q(k) \big\| \leq \beta,\quad \forall k\in\Z_+.
\end{equation}
\end{proposition}

It is possible to strengthen Proposition \ref{prop:simulation} and ensure that we actually have  $X(k)=Q(k)$ for every  $k\in\Z_+$, thus strengthening \eqref{eq: X=Q cond}. However, this stronger result is not needed for our future development, and would require a much more tedious construction of $U(\cdot)$. 
\hide{\comm{Question: Is it guaranteed that $X$ is always non-negative? If yes, it may be worth pointing out. If not, again worth pointing out here.} \oli{Yes. According to \eqref{eq:x-y},
$X(t) =  y(\tf)$; and according to Claim \ref{claim:simul discrete}, $y(t)\in\R_n^+$.}}
It is also worth pointing out here that the perturbed trajectory $X(\cdot)$  in our construction is always non-negative.

Before presenting the detailed proof, we provide some intuition on the issues that arise.
Recall the network evolution rule $Q(k+1)=Q(k)+A({k})+ \big(\RM-I\big) \min\big(\mu(k),Q(k)\big)$ in \eqref{eq:network evolution rule}.
Consider a time $k$ 
at which there is a unique maximizer $\mu(k)$, so that $F(Q(k))$ is a singleton, consisting of the single element $(R-I)\mu(k)$. Suppose furthermore  that
$ \mu(k)\preceq Q(k)$. 
 In this case,  \eqref{eq:network evolution rule} becomes $Q(k+1)=Q(k)+A({k})+ \big(\RM-I\big)\mu(k)$.
Let 
\begin{equation*}
U(t)\,=\,
\begin{cases}
  - (t-k)\big[\big(\RM-I\big)\mu(k) + \lambda\big],  &\textrm{for}\quad t \in (k,k+1),\\
  A(k)-\lambda,& \textrm{for}\quad   t=k+1.\\
\end{cases}
\end{equation*}
In the interval $t\in(k,k+1)$, we  have  $\dot{U}(t)\,= \, -\big[\big(\RM-I\big)\mu(k) + \lambda\big] \,\in\, -F\big(Q(k)\big)-\lambda$. 
Suppose now that
$X(k)=Q(k)$. From the dynamics of $X(\cdot)$ (cf.\ Definition \ref{def:integral pert traj}), we have 
$$\dot{\ptraj}(t)=F\big(X(t)\big)+\lambda + \dot{U}(t)
=F\big(Q(k)\big)+\lambda + \dot{U}(t)
\,=\,
0,$$ 
and the perturbed trajectory remains constant: 
$X(t)=Q(k)$, for $t\in(k,k+1)$. 
Finally, a discontinuity in $U(\cdot)$, at $t=k+1$, forces $X(t)$ to jump to the new value $Q(k+1)$. 
Thus, in this example, we have a perturbed trajectory that agrees with the queue process at integer times.

The above argument will however  fail when $\min\big(\mu(k),Q(k)\big) \ne \mu(k)$, because the received service in time slot $k$, i.e., $\big(\RM-I\big)\min\big(\mu(k),Q(k)\big)$, 
need not belong to $F\big(Q(k)\big)$. 
To circumvent this problem, we find a nearby point $y(k)$ such that $\big(\RM-I\big)\min\big(\mu(k),Q(k)\big) \, \in\, F\big(y(k)\big)$. 
We then construct $\Prt(\cdot)$ 
so that it forces $\ptraj(t)$ to jump to $y(k)$ at time $t=k$, and stay there for $t\in [k,k+1)$.

\begin{proof}
We now provide the detailed proof of Proposition~\ref{prop:simulation}. 
We will be making use of the following known result:
\begin{lemma}[\cite{MannT05} Lemma 5.1]\label{lem:known lemma!}
Given a finite collection  of half-spaces  $H_i\subset \R^n$ with non-empty intersection, there exists a constant $c>0$ such that  
\begin{equation}
 d\Big(x\,,\,\bigcap_i H_i\Big)
 \le c  \cdot
 \max_i d\big(x\,,\,H_i\big)
 ,\qquad
{\forall \ x\in \R^n.}
\end{equation}
\end{lemma}

We begin with a claim.
\begin{claim} \label{claim:simul discrete}
There exists a constant $\kappa$ that satisfies the following. Consider any $x\in\R_+^n$ and any $\mu\in\S(x)$. Let $J = \big\{j\,\big|\, x_j\le \mu_j\big\}$.
Then, there exists a $y\in\R_+^n$, such that $\Ltwo{y-x}\le \kappa$ and 
\begin{equation}\label{eq:DJ in Fy}
\big(\RM-I  \big) D_J(\mu)\subseteq F(y),
\end{equation}
where $D_J(\mu)$ is defined in \eqref{eq:def DJ mu}.
\end{claim}
\begin{proof}[Proof of Claim]
We will leverage Lemma \ref{lem:known lemma!} to find $y$, and then use Claim \ref{claim:fluid} to prove \eqref{eq:DJ in Fy}.
To every $\mu\in\S$, we associate an effective region $\region_\mu$:
\begin{equation} \label{eq:effective region}
\region_\mu = \left\{ z \,\big| \, \mu\in \S(z)    \right\}.
\end{equation}
Fix some $x\in\R_+^n$ and some $\mu\in\S(x)$. The effective region $\region_\mu$ is then the intersection of   half-spaces of the form
\begin{equation}
H_\pi = \left\{  z\in\R^n \,\big|\, z^T \big( I - \RM  \big) \mu \ge z^T \big( I - \RM  \big) \pi   \right\},\qquad \pi\in\S.
\end{equation}
Since $\mu\in\S(x)$, we have $x\in H_{\pi}$ and  $d\big(x,H_\pi\big)=0$, for all $\pi\in \S$.
Let 
\begin{equation}
b \triangleq \max_{\pi\in\S} \, \max_{i\le n } \pi_i
\end{equation}
be the maximum service capacity of any queue over all service vectors.

We also define, for every $i\le n$, two half spaces $H_{j^+} = \big\{z\in\R^n\, \big|\, z_j\ge0 \big\}$ and $H_{j^-} = \big\{z\in\R^n\, \big|\, z_j\le0 \big\}$.
It  follows from the definition of $J$ that for any $j\in J$, $d\big(x,H_{j^-}\big)\le b$, and from $x\in\R_+^n$ that $d\big(x,H_{i^+}\big)=0$, for all $i\le n$.
We define a set $B$, which is determined by the chosen $x\in\R^n_+$ and $\mu\in\S(x)$, as follows:
$$B=\big\{z\in\R^n_+ \cap \region_\mu \mid z_j=0\mbox{ whenever }x_j\leq\mu_j\big\}.$$
Note that the set $B$ is the intersection of finitely many half-spaces of the form $H_{\pi}$, 
$H_{j^+}$, and $H_{j^-}$.
Note furthermore that $B$ contains the origin and  is therefore non-empty.
Finally note that $x$ has a distance of at most $b$ from each of the half-spaces defining $B$. Therefore,
Lemma \ref{lem:known lemma!} implies that 
\begin{equation} \label{eq:cb}
d\big( x,  B  \big) \le c\,b,
\end{equation}
for some constant $c$. In general, the constant $c$ will depend on the particular $x$ and $\mu$ under consideration. Note, however, that the set $B$ is completely determined by $\mu\in\S$ and the set of indices $J=\big\{j\,\big|\, x_j\le \mu_j\big\}$. There are finitely many choices for $\mu$ and for $J$, hence finitely many possible sets $B$. By taking the largest of the constants $c$ associated with different sets $B$, we see that $c$ in \eqref{eq:cb} can be taken to be an absolute constant, independent of $x$ and $\mu$.

Let $y$ be the closest point to $x$ in the set $B$. 
Letting $\kappa=cb$, \eqref{eq:cb} implies that 
\begin{equation}
\Ltwo{y-x}\le \kappa,
\end{equation}
where $\kappa$ is an absolute constant. 

Since $y\in B$, we have $y\in\region_\mu$, so that $\mu\in S(y)$. 
Moreover, for any $j\in J$ i.e., if $x_j\leq \mu_j$, we have  $y_j=0$. 
Let $J'=\{j\mid y_j=0\}$. We then have $J\subseteq J'$.
We now apply this inclusion together with Claim \ref{claim:fluid}, with $y$ and $J'$ playing the role of $q$ and $J$ in the statement of the claim, to obtain 
$$\big(\RM-I\big)D_J(\mu) \,\subseteq\, \big(\RM-I\big)D_{J'}(\mu)
\,\subseteq\, F(y).$$
This completes the proof of the Claim.
\end{proof}

For every $t\in \Z_+$, let $\mu(t)\in\S\big(Q(t)\big)$ be the action  taken  by the scheduler at time $t$, $J(t) = \left\{ j\,\big|\, Q_j(t)\le \mu_j(t)\right\}$, and $y(t){\in\R^n_+}$   be a vector that satisfies $\Ltwo{y(t)-Q(t)}\le \kappa$ and 
\begin{equation}\label{eq:DJ  Q in Fy}
\big(\RM-I  \big) D_J\big(\mu(t)\big)\subseteq F\big(y(t)\big),
\end{equation}
as in Claim \ref{claim:simul discrete}.

We now proceed to the main part of the proof of the proposition. 
With a slight abuse of notation, we will write expressions such as $\sum_{k<t}$ even if $t$ is non-integer, which we will interpret as $\sum_{\{k\in \Z_+ \mid k<t\}}$.
We  define the right-continuous perturbation function $U(\cdot)$ as
\begin{equation} \label{eq:def of U (not right cont) in the simulation lemma}
\begin{split}
&U(t) =\sum_{k\le t -1} \big( A(k)-\lambda \big) \,+\, y(\tf)-Q(\tf)  \\
&\qquad \qquad \, -\,  \tr \Big(\lambda+ \big(\RM-I\big) \min\big( \mu(\tf),Q(\tf)\big) \Big), 
\end{split}
\end{equation}
for all $t\in\R_+$.
 
Let
\begin{equation}\label{eq:def constant c}
b \triangleq \max_{\mu\in\S} \,\sup_{Q\in\R_+^n}  \,\Ltwo{\big(\RM-I\big) \min\big( \mu,Q\big)},
\end{equation}
which is a finite constant. 
For any $t\in \R_+$,
\begin{equation}
\begin{split}
\sup_{\tau\le t} \Ltwo{\Prt(\tau)}\, & \le \, 
 \sup_{\tau\le t} \Ltwo{ \sum_{{k < \tau}} \big(A(k) -\lambda   \big) } \,+\, \sup_{{k \le} \tau} \Ltwo{y(k)-Q(k) }  \,+\, \lVert  \lambda \rVert \,  \\
 \quad  &+\, \sup_{{k \le} \tau} \Ltwo{ \big(\RM-I\big) \min\big( \mu(k),Q(k)\big)}   \\
&\le\,  \sup_{\tau\le t} \Ltwo{ \sum_{{k < \tau}} \big(A(k) -\lambda   \big) } \, +\, \kappa \,+\, \lVert  \lambda \rVert \,+\, b.
\end{split}
\end{equation}	
Hence, \eqref{eq: U bounded cond} follows for $\beta= \kappa + b$.

For every $t\in\R_+$, let 
\begin{equation}\label{eq:x-y}
X(t) =  y(\tf).
\end{equation}
Since $\big\|y(t)-Q(t)\big\|\leq\kappa$, by construction,
the desired equality \eqref{eq: X=Q cond} at integer times is trivially true, with $\kappa$ playing the role of $\beta$. 
It remains to show that $\ptraj(\cdot)$ is a perturbed trajectory.

Let
\begin{equation}
\xi(t) =  \big(\RM-I\big) \min\big( \mu(\tf),Q(\tf)\big)+\lambda,\qquad \forall\ t\in\R_+.
\end{equation}
Since $ \min\big( \mu(\tf),Q(\tf)\big) \in D_{J(\tf)}\big(\mu(\tf) \big)$, it follows from \eqref{eq:DJ  Q in Fy} and \eqref{eq:x-y} that  $\xi(t)\in F\big(\ptraj(t)\big)+\lambda$, for all $t\in \R_+$. 

We  now show that 
\begin{equation}\label{eq:xi}
\ptraj(t) = \int_0^t \xi(\tau)\,d\tau + \Prt(t),\qquad \forall t\geq 0.
\end{equation}
For any $k\in\Z_+$, 
\begin{equation}\label{eq:int = ptraj 1}
\begin{split}
\int_k^{k+1} \xi(\tau)\,d\tau \,+\, \Prt(k+1)- \Prt(k) \, &= \, \Big[\big(\RM-I\big) \min\big( \mu(k),Q(k)\big) 
+ \lambda \Big] \\
&\quad+\, \Big[ A(k)-\lambda + y(k+1)-Q(k+1) \\
&\quad\qquad -\big(y(k)-Q(k)\big)   \Big]
 \\
& =\,\Big[\big(\RM-I\big) \min\big( \mu(k),Q(k)\big) + A(k) \Big]\\
&\quad \,- \big(Q(k+1)-Q(k)\big) + \big(y(k+1)-y(k)\big)    
\\
&=\, y(k+1)-y(k)\\
&=\, X(k+1)-X(k),
\end{split}
\end{equation}
where the third equality follows from the evolution rule \eqref{eq:network evolution rule}.
Moreover, for any $t\not\in\Z$,
\begin{equation} \label{eq:int = ptraj 2}
\begin{split}
\int_{\tf}^{t} \xi(\tau)\,d\tau \,&+\, \Prt(t)- \Prt(\tf) \,  \\
&= \, \tr \Big( \lambda + \big(\RM-I\big) \min\big( \mu(\tf),Q(\tf)\big) \Big) \\
&\quad \ -\,  \tr \Big(\lambda+ \big(\RM-I\big) \min\big( \mu(\tf),Q(\tf)\big) \Big)    \\
&=\, 0\\
&=\, X(t)-X(\tf).
\end{split}
\end{equation} 
Then, a simple induction based on \eqref{eq:int = ptraj 1} and \eqref{eq:int = ptraj 2} implies \eqref{eq:xi},  and therefore $\ptraj(\cdot)$ is a perturbed trajectory of $F(\cdot)+\lambda$ corresponding to $\Prt(\cdot)$, which is the desired result.
\end{proof}


\subsection{Proof of Theorem \ref{th:main sensitivity}} \label{subsec:proof sensitivity th}

Having established a reduction from WMW policies to a MW policy (in Lemma \ref{lem:wmw2mw}), and a reduction from a network, operating under MW policy, to its induced FPCS system (cf. Propositions \ref{prop:simulation fluid} and \ref{prop:simulation}), we can now leverage Theorem \ref{th:sensitivity fpcs} to prove Theorem \ref{th:main sensitivity}.

\medskip
\begin{proof}[Proof of Theorem \ref{th:main sensitivity}]
Consider  a network $\net$ that operates under a $w$-WMW policy. Let  $W=\diag(w)$. Consider a queue length process $Q(\cdot)$  of $\net$ corresponding to an arrival $A(\cdot)$, and a fluid solution $q(\cdot)$ of $\net$ corresponding to an arrival rate vector $\lambda$,  initialized at $q(0)=Q(0)$. 
For each time $t$, let $\tilde{Q}(t) = W^{1/2}Q(t)$, $\tilde{q}(t) = W^{1/2}q(t)$, $\tilde{A}(t)= W^{1/2}A(t)$, and $\tilde{\lambda}= W^{1/2}\lambda$. 
Let $\theta_{\min{}}\triangleq\min_{i}w_i^{1/2}$ and $\theta_{\max{}}\triangleq\max_{i}w_i^{1/2}$. Then, for any time $t\in \Z_+$,
\begin{equation} \label{eq: Qt-qt le Q-q}
\Ltwo{\tilde{Q}(t) - \tilde{q}(t)} \,=\, \Ltwo{W^{1/2} \big(Q(t) - q(t)\big)} \,\ge\, \theta_{\min{}}  \Ltwo{Q(t) - q(t)},
\end{equation}
and
\begin{equation} \label{eq: At-lamt le A-lam}
\Ltwo{\sum_{k\le t}\big(\tilde{A}(k) - \tilde{\lambda}\big)} \,=\, \Ltwo{W^{1/2}\, \sum_{k\le t}\big(A(k) - \lambda\big)} \,\le\, \theta_{\max{}}  \Ltwo{\sum_{k\le t}\big(A(k) - \lambda\big)}.
\end{equation}

Let $\tilde{\net}$ be, as in Lemma \ref{lem:wmw2mw}, a network that operates under a MW policy, and for which $\tilde{Q}(\cdot)$ and $\tilde{q}(\cdot)$ are a queue length process and a fluid solution, respectively. 
Consider the induced  FPCS system  $F$ of $\tilde{\net}$. 
It follows from Proposition~\ref{prop:simulation} that there exists a right-continuous perturbation function $\Prt(\cdot)$ satisfying for any $t\in \R_+$,
\begin{equation} \label{eq: U bounded cond in proof of th}
\sup_{\tau\le t} \Ltwo{\Prt(\tau)}\, \le\, \lVert  \tilde\lambda \rVert + \beta 
+ \sup_{\tau\le t} \Ltwo{ \sum_{k < \tau} \big(\tilde{A}(k) -\tilde\lambda   \big) },
\end{equation}
and a corresponding perturbed trajectory $\ptraj(\cdot)$ of $F(\cdot)+\tilde{\lambda}$ such that 
\begin{equation} \label{eq: X=Q cond in proof of th}
\big\| \ptraj(k) - \tilde{Q}(k) \big\| \leq \beta,\quad \forall k\in\Z_+,
\end{equation}
where $\beta$ is a constant independent of $\tilde{\lambda}$.
Moreover, from Proposition~\ref{prop:simulation fluid}, $\tilde{q}(\cdot)$ is an unperturbed trajectory of $F(\cdot)+\tilde{\lambda}$.
Then, applying Theorem~\ref{th:sensitivity fpcs} for the FPCS sytem $F$, we obtain for any $t\in \R_+$,
\begin{equation} \label{eq:application of Th fpcs to Th1}
\Ltwo{\ptraj(t)- \tilde{q}(t)} \,\le \, \tilde{\divconst}\, \sup_{\tau\le t} \Ltwo{ \Prt(\tau)},
\end{equation}
for some constant $\tilde{\divconst}\ge 1$  that is independent of $\lambda$.
Let $\divconst =  \tilde{\divconst}\, \max(\theta_{\max{}},2\beta)\,/\theta_{\min{}}$. Then, for any $t\in\Z_+$,
\begin{eqnarray*}
\Ltwo{Q(t)-q(t)} &\le & \frac{1}{\theta_{\min{}}}\Ltwo{\tilde{Q}(t)-\tilde{q}(t)}\\
&\le& \frac{1}{\theta_{\min{}}} \Big( \Ltwo{\ptraj(t)- \tilde{q}(t)}  \ +\beta \Big)\\
&\le&  \frac{ 1}{\theta_{\min{}}}\,\Big( \tilde{\divconst}\sup_{\tau\le t} \Ltwo{ \Prt(\tau)} \ +\beta \Big)\\
&\le & \frac{ \tilde{\divconst}}{\theta_{\min{}}}\,\Big( \sup_{\tau\le t} \Ltwo{ \Prt(\tau)} \ +\beta \Big)\\
&\le & \frac{ \tilde{\divconst}}{\theta_{\min{}}} \, \bigg( \lVert  \tilde\lambda \rVert + 2\beta  
+ \sup_{\tau\le t} \Ltwo{ \sum_{{k < \tau}} \big(\tilde{A}(k) -\tilde\lambda   \big) }       \bigg)\\
&\le & \frac{ \tilde{\divconst}}{\theta_{\min{}}} \, \bigg( \theta_{\max{}}\lVert \lambda \rVert + 2\beta  
+\theta_{\max{}} \sup_{\tau\le t} \Ltwo{ \sum_{{k < \tau}} \big(A(k) -\lambda   \big) }       \bigg)\\
&\le & \divconst \, \bigg( 1+\lVert  \lambda \rVert  
+ \sup_{\tau\le t} \Ltwo{ \sum_{{k < \tau}} \big(A(k) -\lambda   \big) }       \bigg),
\end{eqnarray*}
where the relations are due to  \eqref{eq: Qt-qt le Q-q}, \eqref{eq: X=Q cond in proof of th}, \eqref{eq:application of Th fpcs to Th1}, $\tilde{\divconst}\ge 1$, \eqref{eq: U bounded cond in proof of th}, \eqref{eq: At-lamt le A-lam}, and the definition of $\divconst$, respectively.
\end{proof}


\section{\bf Proof of Theorem \ref{th:strong ssc}} \label{sec:proof ssc trans}

In this section we present the proof of Theorem \ref{th:strong ssc}.
Part (a)  is a corollary of Part (b). In the following, we first prove part (b), and then Part (c).

{\bf Proof of Part (b).}  We first  consider a MW policy. 
We then use Lemma \ref{lem:wmw2mw} to extend the result to the case of WMW policies. 
The proof for a MW policy goes along the following lines. We break down a $g(r)$-long interval into subintervals of length $r$. 
We define a ``good'' event $\ee_r$, that the aggregate arrival in each $r$-long interval does not deviate much from its average, and show that this event happens with high probability.
We then use Theorem \ref{th:main sensitivity} to show that  $\ee_r$ implies that the queue length process stays close to a fluid solution, in every subinterval.
These fluid solutions are attracted to $\I{(\lambda)}$ (cf.\ Lemma \ref{lem:absorb const}), and hence also keep the queue length process near $\I{(\lambda)}$.

We now present a detailed proof. 
We fix some $\delta>0$ and some
$\epsilon>0$ such that
\begin{equation}\label{eq:def eps le del/c}
\epsilon  \leq \min\big(\delta,\absorbconst\big)/4\divconst,
\end{equation}
where $\divconst$ is the  sensitivity constant provided by  Theorem \ref{th:main sensitivity}, and  $\absorbconst=\absorbconst(\lambda)$ is the constant in Lemma \ref{lem:absorb const}, associated with $\lambda$.
For  $r\in\N$, we define a good event $\ee_r$:
\begin{equation}
\ee_r
=\Big\{\,  \frac{1}{r} \cdot \sup_{ t\in [ir, (i+1)r)} \Ltwo{\sum_{\tau=ir}^{t}\big( A^r(\tau) -\lambda^r\big)} \le \epsilon ,\quad \forall \, i\in\big[ 0,\,g(r)T/r
\big]
\Big\}.
\end{equation}
We denote the complement of an event $\ee$ by $\ee^c$.
Then, for any $r$,
\begin{equation} \label{eq:prob of eec}
\begin{split}
\probb{\ee_r^c}\,
& =\,\probbB{ \exists\  i\in [ 0,\,g(r)T/r],\ \textrm{s.t. }\ \frac{1}{r} \sup_{ t\in [ir, (i+1)r)} \Ltwo{\sum_{\tau=ir}^{t}\big( A^r(\tau) -\lambda^r\big)} > \epsilon}\\
& \le\, \sum_{i=0}^{\lfloor g(r)T/r\rfloor}\probbb{\frac{1}{r}\sup_{ t\in [ir, (i+1)r)} \Ltwo{\sum_{\tau=ir}^{t}\big( A^r(\tau) -\lambda^r\big)} > \epsilon}\\
& =\, \Big(\big\lfloor g(r)T/r\big\rfloor +1\Big) \,\probbb{\frac{1}{r}\sup_{ t\in[ 0, r)} \Ltwo{\sum_{\tau=0}^{t}\big( A^r(\tau) -\lambda^r\big)} > \epsilon}\\
& =\,\Big(\big\lfloor g(r)T/r\big\rfloor +1\Big)\,  o\Big(1/f\big(r,\epsilon\big)\Big),
\end{split}
\end{equation} 
where the inequality is due to the union bound, the second equality holds because $A^r(\cdot)$ is  a stationary process, and the last equality is because $A^r(\cdot)$, $r\in\N$, is an   $f$-tailed sequence of processes (cf.~Definition \ref{def:f tailed}).
Also note that in the last line, $\epsilon$ is a 
fixed constant, and the $o(\cdot)$ notation is with respect to $r$, as $r$ goes to infinity.

From now on, and since $\lambda$ is fixed, we use the simpler notation $\I$, instead of $\I(\lambda)$. 
Consider an  $r_0\in\N$ such that for every $r\ge r_0$,
\begin{align}
\Ltwo{\lambda^r-\lambda}&\le \divconst\epsilon   \label{eq:lambda r close to lambda}\\
d\big(\qhat^r(0)\,,\, \I\big) &\le 2\divconst\epsilon, \label{eq:cond r_0 q_0} \\
\lVert{\lambda^r}\rVert  + 1 &\le r\epsilon	.\label{eq:cond r_0 largee than lambda}
\end{align}
Such an $r_0$ exists because of the convergence assumptions in the statement of the theorem, which also imply that $\lambda^r$ is a bounded sequence.
For every $r,i\in\Z_{+}$, we define two events, $E_{r,i}$ and $E'_{r,i}$:
\begin{equation}
\begin{split}
E_{r,i}\triangleq \textrm{ the event that }&\,d\Big(Q^r\big(ir\big)\,,\,\I\Big) \le 2\divconst r \epsilon,\\
E'_{r,i}\triangleq \textrm{ the event that }&\,d\Big(Q^r(t)\,,\,\I\Big) \le r\delta ,\quad \forall \, t\in [ir, (i+1)r),
\end{split}
\end{equation}
where $Q^r(\cdot)$ is the queue length process corresponding to the arrival $A^r(\cdot)$.
Using Theorem \ref{th:main sensitivity}, we will now show that for any $r\ge r_0$, $\ee_r$ implies $E_{r,i}$ and $E_{r,i}$, for all $i< g(r)T/r$.
\begin{claim} \label{claim:E and E'}
Fix some
 $r\ge r_0$. The occurrence of the event $\ee_r$ implies the occurrence of the events $E_{r,i}$ and $E'_{r,i}$, for all $i< g(r)T/r$.
\end{claim}
\begin{proof}
The proof is by induction on $i$. For the base case, $E_{r,0}$ follows from \eqref{eq:cond r_0 q_0}, because of $Q^r(0)=r \qhat^r(0)$ and the conic property of $\I$, in \eqref{eq:I is conic}. 
 For the induction step, we will show that for any $i<g(r)T/r$, the events $\ee_r$ and $E_{r,i}$ imply $E_{r,i+1}$ and $E'_{r,i}$.

For any $r\in\Z_+$,  let $q_i^r(t)$ be the fluid solution corresponding to arrival rate $\lambda^r$, and  initialized with $q_i^r(ir)=Q^r(ir)$ at time $ir$. 
Fix an arbitrary $t_0\ge ir$, and let $q(\cdot)$ be a fluid solution corresponding to arrival rate $\lambda$, initialized at $q(t_0) = q_i^r(t_0)$.
From Proposition~\ref{prop:simulation fluid}, $q(\cdot)$  and $q_i^r(\cdot)$ are solutions of $\dot{q}\in F(q)+\lambda$ and $\dot{q}_i^r \in F(q_i^r)+\lambda^r$, respectively, where $F$ is the induced FPCS system of the network.
It then follows from Lemma~4.5 of \cite{Stew11} that for any $t\ge t_0$, $\Ltwo{q_i^r(t)-q(t)}\le (t-t_0) \Ltwo{\lambda^r-\lambda}$. 
As a result,
\begin{equation} \label{eq:d+dt d qir-q}
 \frac{{d^+}}{dt} \Ltwo{q_i^r(t)-q(t)} \,\Big|_{t=t_0}\,\le\, \Ltwo{\lambda^r - \lambda}.
\end{equation}
Suppose that $q_i^r(t_0)\not\in \I$. Then, we also have $q(t_0)\not\in \I$ and Lemma \ref{lem:absorb const} implies that
\begin{equation}
\begin{split}
\frac{{d^+}}{dt} d\big(q_i^r(t),\I\big) \Big|_{t_0} & 
\,\le\, \frac{{d^+}}{dt} \Ltwo{q_i^r(t) - q(t)} \Big|_{t=t_0} \,+\, \frac{{d^+}}{dt} d\big(q(t),\I\big)\Big|_{t=t_0} 
\\
&\,\le\, \Ltwo{\lambda^r-\lambda} -\alpha \\
&\,\le\,  C\epsilon -\alpha\\
&\,\le\,  C\epsilon  -4C\epsilon \\
& \,<\, -2C\epsilon,
\end{split}
\end{equation}
 where the second inequality is from \eqref{eq:d+dt d qir-q}, and the fourth inequality is due to \eqref{eq:def eps le del/c}. 
Moreover, under $E_{r,i}$, we have 
$$d\big(q_i^r(ir),\, \I \big)\,
= d\big(Q^r(ir),\, \I \big)\, \le \, 2Cr\epsilon.$$
 Therefore, under $E_{r,i}$,
\begin{eqnarray}
&& d\big(q_i^r(t),\, \I \big)  \, \le \,  2Cr\epsilon, \qquad \forall\, t\in\big[ir,\, (i+1)r\big), \label{eq:q_ir(t) comes closer to I}\\
&& d\Big(q_i^r\big((i+1)r\big),\, \I \Big) \,=\, 0.    \label{eq:q_ir((i+1)r) in I} 
\end{eqnarray}
Then, for any $r,i\in\Z_+$, and under $E_{r,i}$,
\begin{equation*} \label{eq:ec given e prob}
\begin{split}
d\Big(Q^r\big((i+1)r\big),\, \I\Big)\,
& \le\, \Ltwo{Q^r\big((i+1)r\big)-q_i^r\big((i+1)r\big)} \,+\, d\Big(q_i^r\big((i+1)r\big),\, \I\Big)\\
&= \, \Ltwo{Q^r\big((i+1)r\big)-q_i^r\big((i+1)r\big)}\\
&\le\,  \divconst \bigg(1 \,+\, \lVert \lambda^r\rVert\,+\, \sup_{t\in[ir,(i+1)r)} \,  \Ltwo{\sum_{\tau=ir}^{t} \big(A^r(\tau)\,-\, \lambda^r\big)} \bigg) \\
&\le\, C \big(r\epsilon +  r\epsilon\big)\\
&=\, 2Cr\epsilon,
\end{split}
\end{equation*}
where the second inequality is due to Theorem \ref{th:main sensitivity}, and the last inequality follows from \eqref{eq:cond r_0 largee than lambda} and $\ee_r$. 
This implies $E_{r,i+1}$.
Moreover,
\begin{equation*}\label{eq:e'c given e prob}
\begin{split}
\sup_{ t\in[ir,(i+1)r)}\, d\big(Q^r(t),\, \I\big) \,
&\le\, \sup_{ t\in[ir,(i+1)r)}\,\Big(\Ltwo{Q^r(t)-q_i^r(t)} + d\big(q_i^r(t),\, \I\big)\Big)\\
&\le\, \sup_{ t\in[ir,(i+1)r)}\,\Ltwo{Q^r(t)-q_i^r(t)} \,+\,      2 Cr\epsilon\\
&\le\, \divconst\, \bigg(1 \,+\, \lVert \lambda^r\rVert\,+\,\sup_{t\in[ir,(i+1)r)} \,  \Ltwo{\sum_{\tau=0}^{t-1}\big( A^r(\tau) -\lambda^r\big)} \bigg)\,+\,        2Cr\epsilon \\
&\le\, C \big(r\epsilon +  r\epsilon\big) \,+\,   2C r\epsilon \\
&\le r\delta,
\end{split}
\end{equation*}
where the second inequality is due to \eqref{eq:q_ir(t) comes closer to I}, the third inequality follows from Theorem \ref{th:main sensitivity}, the fourth inequality is from  \eqref{eq:cond r_0 largee than lambda} and $\ee_r$, and the last inequality is due to  the definition of $\epsilon$ in \eqref{eq:def eps le del/c}. This implies $E'_{r,i}$ and completes the proof of the claim.
\end{proof}

Back to the proof of the theorem, let us again fix some $r\ge r_0$. We have
\begin{equation} \label{eq:prob sup le prob eec}
\begin{split}
\prob{\sup_{t\in  [0,T]}\,d\Big(\qhat^r(t)\,,\, \I\Big)> \delta} \,
&=\, \prob{\sup_{t\le  g(r)T}\,d\Big( Q^r(t)/r\,,\, \I\Big)> \delta}\\ 
&=\, \prob{\sup_{t\le  g(r)T}\,d\Big( Q^r(t)\,,\, \I\Big)> r\delta}\\ 
&\le\, \prob{\bigcup_{i\le g(r)T/r  } {E'}_{r,i}^c}\\ 
&\le\, \probb{\ee_r^c},
\end{split}
\end{equation}
where the second equality is due to \eqref{eq:I is conic}, the first inequality is from the definition of $E_{r,i}'$, and the last inequality is due to Claim \ref{claim:E and E'}.
Thus,
\begin{eqnarray*}
 \frac{rf\big(r,\epsilon\big)}{g(r)}\, \prob{\sup_{t\in  [0,T]}\,d\Big(\qhat^r(t)\,,\, \I\Big)> \delta} 
&\le &   \frac{rf\big(r,\epsilon\big)}{g(r)}\, \probb{\ee^c}\\
&\le &  \frac{rf\big(r,\epsilon\big)}{g(r)}\, \Big(\big\lfloor g(r)T/r\big\rfloor +1\Big)\,  o\Big(1/f\big(r,\epsilon\big)\Big)  \\
&\le & o\left(\frac{\big\lfloor g(r)T/r\big\rfloor +1}{g(r)/r}  \right)  \\
&=& o(1) \, \xrightarrow[{\,r\to\infty\,}]{} 0,
\end{eqnarray*}
where the first two inequalities are due to \eqref{eq:prob sup le prob eec} and \eqref{eq:prob of eec}, respectively; and the equality follows from the assumption $\liminf_{r\to\infty} g(r)/r >0$. This completes the proof of Part (b) for the case of a MW policy.

We now present the proof of Part (b) for WMW policies. 
Suppose that a network $\net$ operates under a $w$-WMW policy, and consider an associated network $\tilde{\net}$ as in Lemma \ref{lem:wmw2mw}, along with the variables and processes therein.
It follows from Lemma \ref{lem:wmw2mw} that if the constant function $q(t)=q_0$ is a fluid solution for network $\net$, then the constant function $\tilde{q}(t)=W^{1/2} q_0$ is a fluid solution for $\tilde{\net}$ under a MW policy, and vice versa.
Therefore, $\tilde{\I}=W^{1/2}\I$ is the set of invariant states for $\tilde{\net}$, corresponding to arrival rate $\tilde{\lambda}=W^{1/2}\lambda$. Let $\theta_{\min{}}\triangleq\min_{i}w_i^{1/2}$ and $\theta_{\max{}}\triangleq\max_{i}w_i^{1/2}$. 
Let $\tilde q(\cdot)=W^{1/2}\hat q(\cdot)$, which is the scaled version  of the MW-driven process $\tilde Q(\cdot)$. 
Then, for any $r\in \N$ and any time $t$,
\begin{equation} \label{eq:relation of qhat r and qtilder}
d\Big(\qhat^r(t)\,,\, \I\Big) \, = \, d\Big(W^{-1/2}\tilde{q}^r(t)\,,\, W^{-1/2}\tilde{\I}\Big) \, \le \,\frac{1}{\theta_{\min{}}} d\Big(\tilde{q}^r(t)\,,\, \tilde{\I}\Big).
\end{equation}
In the same vein,
\begin{equation}
\Ltwo{\tilde{A}^r(t)-\tilde\lambda^r} \, = \, \Ltwo{W^{1/2}\big({A}^r(t)-\lambda^r\big)} \,\le\,  \theta_{\max{}} \Ltwo{{A}^r(t)-\lambda^r}.
\end{equation}
As a result, $\tilde{A}^r(\cdot)$, $r\in\N$, is a $\big(\theta_{\max{}} f\big)$-tailed sequence of processes.

As  in \eqref{eq:def eps le del/c}, fix some $\delta>0$ and some
$\epsilon>0$ such that
$
\epsilon  \leq \min\big(\delta\theta_{\min{}} ,\absorbconst\big)/4\divconst,
$
where $\divconst$ is the  sensitivity constant of the network operating under a MW policy (cf. ~Theorem \ref{th:main sensitivity}) and $\absorbconst=\absorbconst(\lambda)$ is the constant in Lemma \ref{lem:absorb const}, associated with $\tilde{\lambda}$. 
Using what we have already established for MW policies, it follows
that 
\begin{equation} \label{eq:collapse of the induced MW system}
 \frac{r\theta_{\max{}}f\big(r,\epsilon\big)}{g(r)}\, \prob{\sup_{t\in  [0,T]}\,d\Big(\tilde{q}^r(t)\,,\, \tilde{\I}\Big)> \delta\theta_{\min{}}}   \xrightarrow[{\,r\to\infty\,}]{} 0.
\end{equation}
This together with \eqref{eq:relation of qhat r and qtilder} implies that
\begin{equation}\label{eq:collapse of MW yields collapse of WMW}
\begin{split}
\frac{rf\big(r,\epsilon\big)}{g(r)}\, &\prob{\sup_{t\in  [0,T]}\,d\Big(\qhat^r(t)\,,\, \I\Big)> \delta}   \\
&\, \le \,
\frac{rf\big(r,\epsilon\big)}{g(r)}\, \prob{\sup_{t\in  [0,T]}\,\frac{1}{\theta_{\min{}}}d\Big(\tilde{q}^r(t)\,,\, \tilde{\I}\Big)> \delta}  \,
\xrightarrow[{\,r\to\infty\,}]{} 0,
\end{split}
\end{equation}
and Part (b) of the theorem follows. 


\medskip
\medskip
{\bf Proof of Part (c).} Throughout this proof we assume that we have fixed a network operated under a MW policy, as well as functions $f(\cdot)$ and $g(\cdot)$ with the properties in the statement of the result, namely, $\lim_{r\to\infty}\, r f\big(r,\epsilon\big)/g(r) = 0$, for all $\epsilon>0$,  and $\lim_{r\to\infty} g(r)/r =\infty$. It is not hard to see that these properties guarantee that there exists
a function $h:\N\to\R_+$  such that
\begin{eqnarray}
\lim_{r\to\infty} \frac{rf\big(r,\delta\big)}{h(r)} &=& 0,\qquad \forall\,\delta>0, \label{eq:rf over h to 0}\\
\lim_{r\to\infty} \frac{h(r)}{g(r)} &=& 0, \label{eq:hr over gr to 0 }\\
\lim_{r\to\infty} \frac{h^2(r)}{r g(r)} &=& \infty.
\end{eqnarray}

Let $\tilde{h}(r)=rg(r)/h(r)$. Then, 
\begin{eqnarray}
\lim_{r\to\infty} \frac{\tilde{h}(r)}{r} = \lim_{r\to\infty} \frac{g(r)}{h(r)} &=& \infty ,  \label{eq:hrr to 0}\\
\lim_{r\to\infty} \frac{\tilde{h}(r)}{h(r)} = \lim_{r\to\infty} \frac{rg(r)}{h^2(r)} &=& 0. \label{eq:hrhr to inf}
\end{eqnarray}

Before continuing with the main part of the proof, we establish
that $\I(\lambda)$ is contained in a low-dimensional subspace. 
The intuition behind this fact is that
$\I(\lambda)$ is  contained in the intersection of different effective regions, each of which is a polyhedron. Recall that $\capregion$ stands for the capacity region of the network.
\begin{claim} \label{claim:I is low dim}
Suppose that $\lambda\in\capregion$ but $\lambda$ is not an extreme point of $\capregion$. Then,
 there exists a nonzero vector $v\in\R^n$ such that $v^T \I(\lambda)=\{0\}$.
\end{claim} 
\begin{proof}
Since $\lambda\in\capregion$, it follows from \eqref{eq:capregion subset convex hull} that $\lambda \in \big( I-\RM\big) \conv(\S) $. Therefore, 
\begin{equation} \label{eq:lambda as a convex comb}
\lambda = \big( I-\RM\big) \sum_{\mu\in \S} \alpha_\mu \mu,
\end{equation}
for some non-negative coefficients $\alpha_\mu$ that sum to one. 
Let us assume that we have fixed one particular set of such coefficients. 

Consider some $x\in \I(\lambda)$, and let $F$ be the induced dynamical system of the network. 
It follows from Proposition \ref{prop:simulation fluid} and the definition of $\I(\lambda)$ that $0\in F(x)+\lambda$.
Therefore, $\lambda\in -F(x)$.
Since $F(x)$ is the convex hull of the vectors $\mu$ that maximize $x^T(I-R)\mu$, we have
\begin{equation}
\lambda = \big( I-\RM\big) \sum_{\nu \in \S(x)} \beta_\nu \nu,
\end{equation}
for some non-negative coefficients $\beta_\nu$ that sum to one. This together with \eqref{eq:lambda as a convex comb} implies that
\begin{equation}
 \big( I-\RM\big) \sum_{\mu\in \S} \alpha_\mu \mu \, =\, \big( I-\RM\big) \sum_{\nu\in \S(x)} \beta_\nu \nu,
\end{equation}
 and as a result,
 \begin{equation} \label{eq:alpha beta conv combs}
 \sum_{\mu\in \S}  \alpha_\mu\, x^T\big( I-\RM\big)  \mu \, =\,  \sum_{\nu \in \S(x)} \beta_\nu\,  x^T \big( I-\RM\big) \nu.
\end{equation}
 Since $\S(x)$ is the set of maximizers of $x^T \big( I-\RM\big) \nu$ over $\nu\in \S$, it follows from \eqref{eq:alpha beta conv combs} that if $\alpha_\mu>0$, then 
 \begin{equation} \label{eq:mu is in the S}
 \mu\in \S(x),
 \end{equation}
and this relation is true for all $x\in\I(\lambda)$. 
This is because otherwise, the left-hand side of \eqref{eq:alpha beta conv combs} would be strictly smaller than the right-hand side. 
 
 On the other hand, since $\lambda$ is not an extreme point of $\capregion$, it follows from \eqref{eq:capregion subset convex hull} that  $\lambda$ is not an extreme point of $\big( I-\RM\big) \conv(\S)$. Then,
 there are at least two service vectors $\mu,\nu\in \S$, for which $\alpha_\mu,\alpha_\nu>0$ and $\big( I-\RM\big)\mu \ne \big( I-\RM\big)\nu$. 
Let $v = \big( I-\RM\big)\big(\mu-\nu\big)$, which is a nonzero vector. 
As already shown in \eqref{eq:mu is in the S}, 
 for any $x\in\I(\lambda)$, we have $\mu,\nu\in \S(x)$.
Therefore, for any $x\in\I(\lambda)$, we have $x^T\big( I-\RM\big)\mu = x^T\big( I-\RM\big)\nu$,  i.e., $v^T x=0$, and the claim follows.
 \end{proof}
 
 
Using the above claim, consider an $(n-1)$-dimensional subspace  $Z$ containing $\I(\lambda)$ and let $w$ be a vector in $\R_+^n\backslash Z$. By suitably scaling $w$,  we can assume that $d\big(w,Z\big) = 1$. Then, $w$ can be decomposed as $w=z+v$ for some $z\in Z$ and some $v\neq 0$ which is orthogonal to $Z$ and has unit norm, $\ltwo{v}=1$.
We also let  
\begin{equation}\label{eq:def b in converse}
b\triangleq\max_{\mu\in\S} \sup_{Q\in \R_+^n} \Ltwo{\big(I-\RM\big) \min\big(\mu,Q\big)},
\end{equation}
which is the maximum instantaneous change in the queue lengths due to service, where $\RM$ and $\S$ are the routing matrix and the set of service vectors, respectively. 
It  follows from \eqref{eq:hrr to 0}  and \eqref{eq:hrhr to inf} that there exists some  $r_0\in\N$ such that for any $r\ge r_0$,
\begin{equation} \label{eq:def of r0 in converse}
2r\delta  + \lVert \lambda\rVert + b \, < \, \tilde{h}(r) . 
\end{equation}

For every $r\in\N$, let $A^r(\cdot)$ be an i.i.d.\ process 
with values 
\begin{equation}\label{eq:def Art for converse}
A^r(t)=\begin{cases}  
\lambda+\tilde{h}(r)\, w,&     \mathrm{w.p.}\quad \frac{1}{h(r)},\\
\lambda   
,&     \mathrm{w.p.}\quad1-\frac{1}{h(r)}.
\end{cases}
\end{equation}
Since $w\in\R_+^n$,  
 $A^r(t)$ is non-negative for all $r$ and $t$, 
 and   from \eqref{eq:hrhr to inf},
\begin{equation}
\E{A^r(t)} \,=\,  \lambda + \frac{\tilde{h}(r)}{h(r)}w\,  \xrightarrow[{\,r\to\infty\,}]{}\, \lambda.
\end{equation}

For any $r\in\N$ and any $T\in\Z_+$, consider an event $E_T^r$:
\begin{equation*}
E_T^r: \textrm{ the event that } A^r(t)=\lambda+\tilde{h}(r)\,w, \textrm{ for at least  one } t\in\big[0,T).
\end{equation*}
Consider some $r\ge r_0$. If $E_r^r$ does not occur, then for any $t< r$ and any $\delta>0$,
\begin{equation*}
\Ltwo{\frac{1}{r}\sum_{\tau=0}^t \big(A^r(\tau)-\lambda\big)} \, =\, 0\,< 
 \, \delta.
\end{equation*}
Therefore,
\begin{equation*}
\begin{split}
 \prob{\frac{1}{r}\sup_{t< r}  \Ltwo{\sum_{\tau=0}^t \big(A^r(\tau)- \lambda \big)} > \delta} \,&\le\, \probb{E_r^r} \,
\le\, \frac{r}{h(r)},
\end{split}
\end{equation*}
where the second inequality follows from \eqref{eq:def Art for converse} and the union bound.
Thus, for any $\delta>0$,
\begin{equation*}
\begin{split}
\lim_{r\to\infty}\,f(r,\delta)\, \prob{\frac{1}{r}\sup_{t< r}  \Ltwo{\sum_{\tau=0}^t \big(A^r(\tau)- \lambda^r\big)} > \delta}\, 
\le \, \lim_{r\to\infty} f(r,\delta)\,\frac{r}{h(r)}
\,=\,0,
\end{split}
\end{equation*}
where the equality is due to \eqref{eq:rf over h to 0}. Hence, $A^r(\cdot)$, $r\in\N$, is an $f$-tailed sequence of processes.

According to our definition of  $v$, we have $v=w-z$, $v$ is orthogonal to the subspace $Z$ containing $\I(\lambda)$, and $\ltwo{v}= d\big(w,Z\big)= 1$. 
Therefore, for any $r\ge r_0$, if $A^r(t)=\lambda + \tilde{h}(r)\,w$ for some  $t$,  then
\begin{eqnarray*}
v^T\,\Big( Q^r(t+1)-Q^r(t)\Big)
&= &  v^T\,A^r(t) \,+\,v^T \big(\RM-I\big) \min\big(\mu(t),Q(t)\big)\\
&\ge &  v^T\,A^r(t) \,- \, b\,\lVert v\rVert \\
& = &  \tilde{h}(r)  \, +\, v^T\lambda \,- \, b   \\
& \ge &   \tilde{h}(r)  \, - \, \lVert \lambda \rVert \,- \, b \\
& > & 2 r \delta,
\end{eqnarray*}
where the inequalities are due to \eqref{eq:def b in converse}, $\lVert v\rVert =1$, and \eqref{eq:def of r0 in converse}, respectively. 
Now recall that $v$ is orthogonal to  the subspace $Z$ containing $\I(\lambda)$. Whenever we have $A^r(t)=\lambda + \tilde{h}(r)\,v$, we have a jump of size at least $2r\delta$ in a direction orthogonal to $\I(\lambda)$, from which it is not hard to see that 
\begin{equation}
\max \bigg( d\Big(Q^r(t+1),\,\I(\lambda)\Big) \,,\, d\Big(Q^r(t),\,\I(\lambda)\Big)\bigg)>\, r\delta.
\end{equation}
This implies that for any $r\ge r_0$, if the event $E_{g(r)T}^r$ occurs, then $$\sup_{t\in  [0,g(r)T]}\,d\Big(Q^r(t)\,,\, \I(\lambda)\Big)>r\delta.$$ Therefore,
\begin{eqnarray*}
\lim_{r\to \infty}\, \prob{\sup_{t\in  [0,T]}\,d\Big(\qhat^r(t)\,,\, \I(\lambda)\Big)>\delta} 
& = &  \lim_{r\to \infty}\, \prob{\sup_{t\in  [0,g(r)T]}\,d\Big(Q^r(t)\,,\, \I(\lambda)\Big)>r\delta} \\
&\ge & \lim_{r\to \infty}\, \probb{ E_{g(r)T}^r }\\
&\ge & \lim_{r\to \infty}\, \Big( 1\,-\, \big( 1- {1}/{h(r)}\big) ^{g(r) T} \Big) \\
&\ge & \lim_{r\to \infty}\, \Big( 1\,-\, e^{-g(r)T/h(r)}  \Big) \\
&=& 1,
\end{eqnarray*}
where the last equality is due to \eqref{eq:hr over gr to 0 }. This completes the proof of Part (c).

\medskip
\noindent
{\bf Remark.} The requirement, in Part (c) of the theorem, that $\lambda$ is not an extreme point of $\capregion$ cannot be removed. For a trivial example, consider a single queue with 
a single (one-dimensional) service vector $\mu=1$. Then, $\capregion=[0,1]$.
If $\lambda=1$, then every $q\geq 0$ is an invariant state: $\I(1)=[0,\infty)$. The conclusion of Claim \ref{claim:I is low dim} fails to hold and it is certainly impossible for the state to be outside $\I(1)$.

\medskip
\noindent
{\bf Remark.} The process $A^r(t)$, used in the proof of Part (c) is not uniformly bounded, as it can have bursts of size $\tilde{h}(r)$.
With some additional effort, and a slightly more complicated proof, 
it is possible to carry out a construction under which each component of $A^r(t)$ is bounded by some constant, independent of $r$ or $t$. The basic idea is that having excess arrivals (but of bounded size) over a time period of length  $O(\tilde{h}(r))$ has an  effect comparable to a single burst of size $O(\tilde{h}(r))$.

\hide{
Therefore, for any $r\ge r_0$, $E_{0,i}$ follows from \eqref{eq:cond r_0 q_0}. Then,
\begin{equation}\label{eq:bound the prob if union of e'c}
\begin{split}
\prob{\sup_{t\in  [0,T]}\,d\Big(\qhat^r(t)\,,\, \I\Big)> \delta} \,
&=\, \prob{\sup_{k\le  g(r)T}\,d\Big(\tfrac{1}{r} Q^r(k)\,,\, \I\Big)> \delta}\\ 
&=\, \prob{\sup_{k\le  g(r)T}\,d\Big( Q^r(k)\,,\, \I\Big)> r\delta}\\ 
&\le\, \prob{\bigcup_{i\le g(r)T/r  } {E'}_{r,i}^c}\\ 
&\le\, \prob{ \bigg\{\bigcup_{i\le g(r)T/r  } {E}_{r,i}^c\bigg\}\,\,\bigcup\,\, \bigg\{\bigcup_{i\le g(r)T/r  } {E'}_{r,i}^c\bigg\}}\\
&=\, \prob{ \bigg\{\bigcup_{i\le g(r)T/r  } {E}_{r,i}^c\bigg\}\,\,\bigcup\,\, \bigg\{\bigcup_{i\le g(r)T/r  } \left\{{E'}_{r,i}^c \bigcup E_{r,i}\right\}\bigg\}}\\
&=\, \prob{ {E}_{r,0}^c  \,\bigcup\,\bigg\{\bigcup_{i\le g(r)T/r  } \left\{ {E}_{r,i+1}^c  \bigcap E_{r,i}\right\}\bigg\}\,\,\bigcup\,\, \bigg\{\bigcup_{i\le g(r)T/r  } \left\{{E'}_{r,i}^c \bigcap E_{r,i}\right\}\bigg\}}\\
&\le\, \prob{E_{r,0}^c}\, + \,\sum_{i=0}^{ g(r)T/r  } \bigg(   \prob{E_{r,i+1}^c \bigcap E_{r,i}} \,+\, \prob{E_{r,i}'^c \bigcap E_{r,i}}    \bigg)\\
&=\, \sum_{i=0}^{ g(r)T/r  } \bigg(   \prob{E_{r,i+1}^c  \,\big| \, E_{r,i}} \,+\, \prob{E_{r,i}'^c  \,\big| \, E_{r,i}}    \bigg)\,\prob{ E_{r,i}}\\
&\le \sum_{i=0}^{ g(r)T/r  } \bigg(   \prob{E_{r,i+1}^c  \,\big| \, E_{r,i}} \,+\, \prob{E_{r,i}'^c  \,\big| \, E_{r,i}}    \bigg)\\
\end{split}
\end{equation}
where the second equality is due to \eqref{eq:I is conic}, the first inequality is from the definition of $E_{r,i}'$, third inequality is an application of the union bound, and the last equality is because $E_{r,0}$ holds for all $r\ge r_0$. 
We now elaborate on bounding the two terms in the above summation.

For any $r,i\in\Z_+$,
\begin{equation}\label{eq:e'c given e prob}
\begin{split}
\prob{E'_{r,i}\,\big| \, E_{r,i}} \,&\,= \prob{\sup_{ t\in[ir,(i+1)r)}\, d\big(Q^r(t),\, \I\big)\,\le\, r\delta \,\bigg| \, E_{r,i}}\\
&\ge\, \prob{\sup_{ t\in[ir,(i+1)r)}\,\Big(\Ltwo{Q^r(t)-q_i^r(t)} + d\big(q_i^r(t),\, \I\big)\Big)\,\le\, r\delta \,\bigg| \, E_{r,i}}\\
&\ge\, \prob{\sup_{ t\in[ir,(i+1)r)}\,\Big(\Ltwo{Q^r(t)-q_i^r(t)} + d\big(q_i^r(ir),\, \I\big)\Big)\,\le\, r\delta \,\bigg| \, E_{r,i}}\\
&\ge\, \prob{\sup_{ t\in[ir,(i+1)r)}\, \Ltwo{Q^r(t)-q_i^r(t)}\,\le\, r\delta - \divconst r \epsilon \,\bigg| \, E_{r,i}}\\
&\ge\, \prob{\sup_{ t\in[ir,(i+1)r)}\, \Ltwo{Q^r(t)-q_i^r(t)}\,\le\, \divconst r \epsilon \,\bigg| \, E_{r,i}}\\
&\ge\, \prob{\sup_{t\in[1,r)} \,  \Ltwo{\sum_{\tau=0}^{t-1}\big( A^r(\tau+ir) -\lambda^r\big)} \le\,  r\epsilon \,\bigg| \, E_{r,i}}\\
&= \, \prob{\sup_{t\in[1,r)} \,  \Ltwo{\sum_{\tau=0}^{t-1}\big( A^r(\tau+ir) -\lambda^r\big)} \le\,  r\epsilon }\\
&= \, 1-o\Big(1/f\big(r,\delta\big)\Big),
\end{split}
\end{equation}
where  the second inequality is due to\eqref{eq:q_ir(t) comes closer to I}, the fifth inequality is an application of Theorem \ref{th:main sensitivity}, the last equality is because  $A^r(\cdot)$ is stationary, and the last  inequality is by definition (\ref{eq:def f-tailed arrival}).  \oli{Is the last equality true?}

If $E_{r,i}$ is true, then $d\big(q_i^r(ir),\, \I\big)\le\divconst r\epsilon<r\absorbconst$. In this case, it follows from  Lemma \ref{lem:absorb const} that $q_i^r\big((i+1)r\big)\in \I$. Then,
\begin{equation} \label{eq:ec given e prob}
\begin{split}
\prob{E_{r,i+1}\,\big| \, E_{r,i}} \,&\ge \prob{d\Big(Q^r\big((i+1)r\big),\, \I\Big)\le \divconst r\epsilon \,\bigg| \, E_{r,i}}\\
&\ge \prob{\Ltwo{Q^r\big((i+1)r\big)-q_i^r\big((i+1)r\big)} + d\Big(q_i^r\big((i+1)r\big),\, \I\Big)\le \divconst r\epsilon \,\bigg| \, E_{r,i}}\\
&=\probB{\Ltwo{Q^r\big((i+1)r\big)-q_i^r\big((i+1)r\big)}\le \divconst r\epsilon  \,\bigg| \, E_{r,i} }\\
&\ge \prob{\sup_{t\in[1,r]} \,  \Ltwo{\sum_{\tau=0}^{t-1} A^r(\tau+ir)\,-\, t\lambda^r} \le\,  r\epsilon  \,\bigg| \, E_{r,i}}\\
&\ge \prob{\sup_{t\in[1,r]} \,  \Ltwo{\sum_{\tau=0}^{t-1} A^r(\tau+ir)\,-\, t\lambda^r} \le\,  r\epsilon }\\
&= 1-o\Big(1/f\big(r,\delta\big)\Big).
\end{split}
\end{equation}
 
Plugging \eqref{eq:e'c given e prob}   and \eqref{eq:ec given e prob}   in \eqref{eq:bound the prob if union of e'c}, we obtain
\begin{equation}
\begin{split}
\frac{rf\big(r,\delta\big)}{g(r)}\, \prob{\sup_{t\in  [0,T]}\,d\Big(\qhat^r(t)\,,\, \I\Big)> \delta} 
&\le \, \frac{rf\big(r,\delta\big)}{g(r)}\, \sum_{i=0}^{ g(r)T/r  } \bigg(   \prob{E_{r,i+1}^c  \,\big| \, E_{r,i}} \,+\, \prob{E_{r,i}'^c  \,\big| \, E_{r,i}}    \bigg)\\
&= \, \frac{rf\big(r,\delta\big)}{g(r)} \, \sum_{i=0}^{ g(r)T/r  } \bigg(   o\Big(1/f\big(r,\delta\big)\Big)\,+\,o\Big(1/f\big(r,\delta\big)\Big)  \bigg)\\
&= \, \frac{rf\big(r,\delta\big)}{g(r)}  \,\,\big\lceil g(r)T/r\big\rceil \, o\Big(1/f\big(r,\delta\big)\Big)\\
&= \, o(1) \, \xrightarrow[{\,r\to\infty\,}]{} 0,
\end{split}
\end{equation}
where the last equality follows from the assumption $g(r)=\Omega(r)$.
This completes the proof of the theorem.
 }


\section{Discussion} \label{sec:discussion}


In this section we review our main results, their implications, and directions for future research.

\subsection{Main Results}
We have established a deterministic bound on the sensitivity of  queue length processes with respect to arrivals, under a Max-Weight policy.
In particular, 
we showed that the distance between a queue length process and a fluid solution remains bounded by a constant multiple of the deviation of the aggregate arrival process from its average. 
The  bound allows for tight approximations of the queue lengths in terms of fluid solutions, which are much easier to analyse, and leads to a simple derivation of a fluid limit result; cf.\ Corollary~\ref{c:fluid}. 
We then exploited this sensitivity result to prove matching upper and lower bounds for the time scale over which  additive SSC occurs  under a MW policy. 
As a corollary, we established strong (additive) SSC of MW dynamics in diffusion scaling under conditions more general than previously available.

For the case of i.i.d.\ arrivals, we established additive SSC 
over time intervals whose length scales exponentially with $r$.
Such a result could also be proved with a more elementary argument, by viewing the distance from the invariant set as a Lyapunov function and using the drift properties that we established; cf.\ Lemma~\ref{lem:absorb const}.
Nonetheless, such Lyapunov-based approaches are hard to generalize to broader classes of arrival processes.  
In contrast, our sensitivity bounds in Theorem \ref{th:main sensitivity} allow for the arrival processes to be arbitrary and yield strong approximation results as long as the driving process has some reasonable concentration properties.

\subsection{Other Applications and Extensions}
\hide{An immediate consequence of our sensitivity result is the closeness of a fluid-scaled process ($\hat{q}(t) =\frac{1}{r}Q(\lfloor rt\rfloor) $) to a fluid solution. One such bound is given in \cite{ShahW12} Theorem 4.3: $d\big(\hat{q}^r(t), \I  \big)\to 0$ as $\max_{k< t} \frac1r \Ltwo{\sum_{\tau\le k} \big(A^r(\tau)-\lambda\big)}\to 0$. Our sensitivity bound further implies that the rate of convergence in the former limit is upper bounded by the convergence rate of the latter.}

A similar sensitivity bound can also be proved, using the same line of argument,  for continuous time networks e.g., with Poisson arrivals, operating under a MW policy.  Similarly, 
 for a more general class of stochastic processing networks, and under Assumption 1 of \cite{DaiL05}, it is not hard to see that the fluid dynamics 
will again be a  subgradient flow, and that our results can be extended to such systems.

In the same spirit, we believe that the results, including additive SSC, can be extended to the case of backpressure policies\footnote{ Backpressure policies are extensions of the MW policy, in which routing is no longer fixed. 
In particular, there is  a fixed set of service vectors, where each service vector $\mu$ associates a rate $\mu_{ij}$ to each link $ij$. A backpressure policy then chooses at each time a service vector $\mu$ that maximizes $\sum_{ij} \mu_{ij} (Q_i-Q_j)$, where the sum is taken over all links $ij$.} \cite{TassE92,GeorNT06,Neel10}, for networks in which the routing is no longer fixed, and where the different vectors $\mu$ determine the set of links to be activated.

Another direction concerns SSC results in steady-state, that is, the extent to which
 the steady-state distribution will be concentrated in a neighbourhood of the invariant set.  
 Following a Lyapunov-based approach, several works \cite{MaguS15,MaguBS16,ShahTZ16} have proved exponential tail bounds for the steady-state distribution, for the case of i.i.d.\ arrivals. We believe that Theorem \ref{th:main sensitivity} provides an approach for establishing similar bounds  for the case of non-i.i.d.\ and non-Markovian arrivals.

Finally, another problem where the fluid model turns out to be analytically beneficial  concerns delay stability under a MW policy in the presence of heavy-tailed traffic. The references \cite{MarkMT16} and \cite{MarkMT18} studied the question whether a certain queue has finite expected delay (``delay stability'') in the presence of other queues that are faced with heavy-tailed arrivals.
They provided a necessary condition for this to be the case, in terms of certain properties of the associated fluid model, and raised the question whether 
under some assumptions, this condition is also sufficient.
Using the sensitivity results of the current paper, we 
are able to resolve a variant of this question 
as will be reported in  a forthcoming publication.

\old{
\subsection{Extensions to Backpressure Policies}
The model we considered in this paper concerns MW policies with fixed routing, determined in terms of a routing matrix $\RM$.
However, our results are valid also for variants of MW policies with unfiexd (or backpressure) routing, under which the real time routing decisions made by the MW policy rely on the current backlogs in the network \cite{TassE92,GeorNT06,Neel10}.
More specifically, consider a network comprising a number of 
traffic sessions corresponding to packets with certain destinations 
and a set of link activation actions $\S$, where each $\mu\in\S$ determines for every link $ij$ a rate $\mu_{ij}$ that can be transmitted on that link upon activation of $\mu$.
A MW policy with backpressure routing chooses at time $t$ an action $\mu\in\S$ that maximizes 
\begin{equation}
\sum_{\textrm{links }ij} \max_{s} \big(Q_i^s(t) - Q_j^s(t)\big) \mu_{ij},
\end{equation}
where $Q_i^s(\cdot)$ is the queue length at node $i$ corresponding to session $s$ (equivalently the the share of backlog in node $i$ destined towards node $s$); and maximum ranges over all sessions.
Upon activation of $\mu\in\S$, for any link $ij$, a packet of size $\min\big(Q_i^{s^*}(t),\mu_{ij}\big)$ leaves $Q_i^{s^*}$ and joins  $Q_j^{s^*}$, where  $s^*$ is a maximizer of $Q_i^s(t) - Q_j^s(t)$. 
\oli{This should suffice to capture the situation when queues run out of packets.}

It is not difficult to show via a straightforward extension that Theorems \ref{th:main sensitivity} and \ref{th:strong ssc} are valid under the above variant of MW policy with backpressure routning.
The proof of the sensitivity result (i.e., Theorem~\ref{th:main sensitivity}) is essentially the same, except for some of the details in the proofs of Propositions~\ref{prop:simulation fluid} and \ref{prop:simulation}. 
\oli{Acually the definition of the fluid model in \eqref{eq:fluid diff eq 1}--\eqref{eq:fluid diff eq 6} would change.}
The prood of the SSC results (i.e., Theorem~\ref{th:strong ssc} and its corollaries)  do not depend on the details of routing considered, and would work in their current form.}

\subsection{Open Problems}
The sensitivity bound in Theorem \ref{th:main sensitivity}  applies only to  MW and WMW policies. 
It is not clear whether a similar bound holds for the more general classes of MW-$\alpha$ and MW-$f$ Policies.  
A similar question also arises about SSC: does Theorem \ref{th:strong ssc} hold under a MW-$\alpha$ policy?

\hide{
\oli{Is it good to give an analogue theorem for sensitivity of  continuous time networks? \\
What about relaxing the fixed routing assumption?}}


\medskip
\medskip

\appendix


\section{Proof of Lemma \lowercase{\ref{lem:absorb const}}}  \label{ap:a}
In this appendix, we present the proof of Lemma \ref{lem:absorb const}. The high level idea is that $q(\cdot)$  is a trajectory of a subgradient dynamical system corresponding to a convex function that has $\I(\lambda)$ as its  set of minimizers. We then show that in such a system, all trajectories are attracted to the set of minimizers, at a uniform rate.

Consider the convex function $\Phi$ in \eqref{eq:def phi for fpcs},
\begin{equation*}
\Phi(x)=\max_{\mu\in \S}\, \big((I-\RM)\mu\big)^Tx,
\end{equation*}
and let $F=-\partial \Phi$ be its subgradient  field. Then, $F$ is the induced FPCS system of the network (cf. Definition \ref{def:induced net}), and Proposition~\ref{prop:simulation fluid} states that 
fluid solutions are the same as the non-negative unperturbed trajectories of $\dot{q}\in F(q)+\lambda$. 

We define another convex function $\Phi_\lambda$ by $\Phi_\lambda(x) = \Phi(x)-\lambda^Tx $, for all $x\in\R^n$.
Since $\lambda$ is in the capacity region, \eqref{eq:capregion subset convex hull} implies that $\lambda\in \big( I-\RM\big) \conv(\S)$. Then, for any $x\in \R^n$,
\begin{equation}\label{eq:phi non neg}
\begin{split}
\Phi_\lambda(x)\, &= \, \Phi(x) -\lambda^Tx\\
&=\, \max_{\mu\in\S }   \big((I-\RM) \mu   \big)^Tx -\lambda^Tx\\
&\ge\, \lambda^Tx-\lambda^Tx \\
&=\, 0.
\end{split}
\end{equation} 
On the other hand, we have $\Phi_\lambda(0)=0$. It follows that the set of minimizers of $\Phi_\lambda$, denoted by $\Gamma$, is
\begin{equation} \label{eq:omega cap orthant = I}
 \Gamma = \left\{x\in\R^n\, \big|\, \Phi_\lambda(x)= 0\right\}.
\end{equation}

We use the shorthand notation $\I$ for $\I(\lambda)$.
We now develop a characterization and a simple polyhedral description of the set $\I$.
Recall that  $\I$ is the set of all $x_0\in\R^n_+$ for which the constant trajectory $q(t)=x_0$ is a fluid solution. 
As pointed out earlier, 
fluid solutions are the same as the non-negative unperturbed trajectories of the system
$F(\cdot)+\lambda$, where $F +\lambda=-\partial\Phi_\lambda$. 
In particular,  $q(t)=x_0$ is a constant trajectory of this system if and only if 
$\partial \Phi_\lambda(x_0)$ contains the zero vector, which is the case if and only if $x_0$ is a minimizer of
$\Phi_\lambda$, i.e., $x_0\in\Gamma$.
Therefore, 
\begin{equation}\label{eq:omega = I}
\I\, =\,\Gamma \cap \R_+^n .
\end{equation}

For any $\mu\in\S$, we define the half-space
$H_\mu=\big\{x\in\R_n\,|\, \big((I-\RM)\mu \big)^Tx\le \lambda^Tx\big\}$. We also let, for $i=1,\ldots,n$, $H_i=\big\{x\in\R_n\,|\, x_i\ge 0\big\}$. 
We will now show that
\begin{equation}\label{eq:I is the intersection of Hs}
\I\, =  \, \Big(\bigcap_{\mu\in\S} H_\mu\Big)\,  \bigcap \, \Big(\bigcap_{i=1}^{n} H_i \Big).
\end{equation}
Suppose that $x_0\in\I$. Then, $x_0\in\R^n_+$, i.e., $x_0\in H_i$, for all $i$. Furthermore, 
$x_0\in\Gamma$ and, from \eqref{eq:omega cap orthant = I},  $\Phi_\lambda(x_0)=0$. From 
\eqref{eq:phi non neg}, this implies that $\max_{\mu\in\S}\big((I-R)\mu\big)^Tx_0=\lambda^Tx_0$, so that $\big((I-R)\mu\big)^Tx_0\leq \lambda^Tx_0$, or equivalently $x_0\in H_\mu$, for all $\mu$.
This argument can be reversed. If $x_0$ belongs to all of the half-spaces $H_{\mu}$, we have $\Phi_\lambda(x_0)\leq 0$, which in light of \eqref{eq:phi non neg} implies that $\Phi_\lambda(x_0)=0$, or $x_0\in\Gamma$. If in addition $x_0\in\R^n_+$, then, from 
\eqref{eq:omega = I}, we obtain $x_0\in\I$. This concludes the proof of \eqref{eq:I is the intersection of Hs}.

Having characterized the set $\I$, we now turn our attention to the dynamics that drive trajectories towards $\I$. Fix some 
$\mu\in\S$.   
For any $x\not\in H_\mu$, 
the closest point to $x$ in $H_\mu$ lies on the boundary of $H_\mu$, i.e., on the subspace 
$W_\mu=\left\{ z\in\R^n\,\big|\,   g_\mu ^Tz= 0 \right\}$,
where $g_\mu= (I-\RM)\mu - \lambda$. 
For any $x\in\R^n$ and $\mu\in\S$, either  $d(x,H_\mu)=0$ (which includes the case  where $g_\mu=0$, so that $W_\mu=H_\mu=\R^n$),  
or $g_\mu\neq 0$ and $x\not\in H_\mu$, in which case
\begin{equation}
d\big(x,\, H_\mu\big)\,=\, d\big(x,\, W_\mu\big)\, =\,  \frac{g_\mu^T x}{\Ltwo{{g_\mu}}},
\end{equation}
which is the length of the projection of $x$ on the normal vector to $W_\mu$. 
Thus, in both cases, we have 
\begin{equation} 
 {\|g_\mu\|} \cdot d\big(x,\, H_\mu\big)\,=\, \max\left( g_\mu^T x,\, 0 \right).  
\end{equation}
Then, for any $x\in\R^n$, 
\begin{eqnarray*} 
\Phi_\lambda(x) &=& \max_{\mu\in \S} \big((I-\RM)\mu - \lambda\big)^T x\\
&=& \max_{\mu\in \S} g_\mu^T x\\
&=&  \max_{\mu\in \S}  \,\max\big(  g_\mu^T x,\, 0 \Big)\\
&=&  \max_{\mu\in \S} \Ltwo{g_\mu}  \cdot  d\big(x,\, H_\mu\big).
\end{eqnarray*}
Let $\epsilon = \min \big\{  \|{g_\mu}\| \, :\, \mu\in\S,\,  g_\mu\neq 0  \big\}$. 
Without loss of generality, we assume that $\epsilon>0$; otherwise, we would be dealing with a trivial system where every $g_\mu$ is zero, and $\Phi_\lambda$ is identically zero.
Note that for any $x\in \R_+^n$, we have
\begin{equation} \label{eq:phi ge eps dH}
\begin{split}
\Phi_\lambda(x) &\, \ge \, \epsilon \max_{\mu\in\S } d\big(x\,,\,H_\mu\big) \\
&\,=\,  \epsilon \max\Big(  \max_{\mu\in\S } d\big(x\,,\,H_\mu\big)\,,\, \max_{i=1,\ldots,n } d\big(x\,,\,H_i\big)\Big),
\end{split}
\end{equation}
where the equality is because when $x\in \R_+^n$, we have $d\big(x\,,\,H_i\big)=0$, for all $i$.

Consider a fluid solution $x(\cdot)$, and 
fix a time $t$ such that $x(t)\not\in \I$. Let $x_0 = x(t)$ and let $z$ be the element of $\I$ that is closest  to $x_0$.  Then, at time $t$, 
\begin{eqnarray*}
\frac{d}{dt} {d\big(x(t),\,\I\big)}
&=& \lim_{h{\downarrow} 0} \frac{d\big(x(t+h),\,\I\big)\,-\, d\big(x(t),\,\I\big)}{h}\\
&=& \lim_{h{\downarrow} 0} \frac{d\big(x(t+h),\,\I\big)\,-\, d\big(x(t),\,z\big)}{h}\\
&\le & \lim_{h{\downarrow} 0} \frac{d\big(x(t+h),\,z\big)\,-\, d\big(x(t),\,z\big)}{h}\\
&=& \frac{d^+}{dt} {d\big(x(t),\,z\big)}\\
&=&  \frac{\big(x(t)-z\big)^T \dot{x}(t)}{\Ltwo{x(t)-z}}\\
&=&  \frac{\big(x_0-z\big)^T \dot{x}(t)}{d\big(x_0,\,\I\big)}.
\end{eqnarray*}
Putting everything together, we obtain
\begin{eqnarray*}
\frac{d}{dt} {d\big(x(t),\,\I\big)} 
&\le &  \frac{\big(x_0-z\big)^T \dot{x}(t)}{d\big(x_0,\,\I\big)}\\
&\le & \frac{\Phi_\lambda(z)-\Phi_\lambda(x_0)}{d\big(x_0,\,\I\big)}\\
& =& -\frac{\Phi_\lambda(x_0)}{{d\big(x_0\,,\,\I\big)}}\\
&\le & -\,  \frac{    \epsilon \,\max\Big(  \max_{\mu\in\S } d\big(x_0\,,\,H_\mu\big)\,,\, \max_{i=1,\ldots, n } d\big(x_0\,,\,H_i\big)\Big)    }{d\Big(x_0\,,\, \Big(\bigcap_{\mu\in\S} H_\mu\Big)\,  \bigcap \, \Big(\bigcap_{i=1}^{n} H_i \Big) \Big)}\\
& \le & -{c\,\epsilon };
\end{eqnarray*}
where the second inequality above is because $-\dot{x}(t)$ is a subgradient of $\Phi_\lambda$ at $x_0$; the equality is 
because $z\in\I$ and $\Phi_\lambda$ vanishes on $\I$, by
\eqref{eq:omega = I}, \eqref{eq:omega cap orthant = I}; the third inequality is due to \eqref{eq:phi ge eps dH} and \eqref{eq:I is the intersection of Hs}; and the last inequality follows from Lemma \ref{lem:known lemma!}, for the constant $c$ therein. 
This completes the proof of Lemma \ref{lem:absorb const}, with
 $\absorbconst(\lambda) = c\, \epsilon  $.

\medskip

\section{Proof of Lemma \lowercase{\ref{lem:iid tail}}} \label{app:proof lem iid tail}
Let $X_1,\ldots,X_t$ be i.i.d.\ random variables, taking values in $[0,a]$, and let 
$\bar{X}=\big(X_1+\cdots+X_t\big)/t$. Then, for any $\delta>0$, Hoeffding's  inequality \cite{Hoef63} yields
\begin{equation}   \label{eq:hoeffding}     
\probb{\left\lvert\bar{X}-\E{\bar{X}}\right\rvert>\delta} \le 2 \exp\left(-\frac{2t\delta^2}{a^2}  \right).
\end{equation}
For any fixed  $r\in\N$ and  $t\le r$, 
\begin{equation}\label{eq:first use of hoef}
\begin{split}
\prob{\frac{1}{r} \Ltwo{\sum_{\tau=0}^t \big(A^r(\tau)- \lambda^r\big)} > \delta}\, 
& \le\, \sum_{i=1}^n \prob{\frac{1}{r}\big\lvert{\sum_{\tau=0}^t \big(A_i^r(\tau) - \lambda_i^r\big)}\big\rvert > \frac{\delta}{\sqrt{n}}}  \\
&=\, \sum_{i=1}^n \prob{\frac{1}{t+1} \big\lvert{\sum_{\tau=0}^t \big(A_i^r(\tau) - \lambda_i^r\big)}\big\rvert > \frac{r\delta}{(t+1)\sqrt{n}}}\\
&\le\, 2n \exp\left(-\frac{2 \big(t +1\big)}{a^2}\left(\frac{r\delta}{\big(t+1\big) \sqrt{n}}\right)^2 \right)\\
& \le\, 2n \exp\left(-\left(\frac{2r}{r+1}\right)\frac{ r\delta}{na^2} \right),
\end{split}
\end{equation}
where the first inequality holds because if the Euclidean norm is above $\delta$, then at least  one of the components must be above $\delta/\sqrt{n}$, together with the union bound, 
and the  third inequality is due to (\ref{eq:hoeffding}). 
As in the statement of the lemma, let $\beta\in(0,2)$ and $f(r,\delta)=\exp\big(\beta r\delta/n a^2\big)$. We then have
\begin{equation}
\begin{split}
f\big(r,\delta\big)\,&\prob{\frac{1}{r}\sup_{t\le r} \Ltwo{\sum_{\tau=0}^t \big(A^r(\tau)- \lambda^r\big)} > \delta}\\
&\le\,f\big(r,\delta\big)\,\sum_{t\le r} \prob{\frac{1}{r} \Ltwo{\sum_{\tau=0}^t \big(A^r(\tau)- \lambda^r\big)} > \delta}\\
&\le\, f\big(r,\delta\big)\, 2nr \exp\left(-\left(\frac{2r}{r+1}\right) \frac{r\delta}{na^2} \right)\\
&= \, 2nr\exp\left(\frac{\beta r\delta}{na^2}-\left(\frac{2r}{r+1}\right) \frac{r\delta}{na^2} \right)\, \xrightarrow[{\,r\to\infty\,}]{}\, 0,
\end{split}
\end{equation}
where   the second inequality is due to (\ref{eq:first use of hoef}), and the last implication is because $\beta<2$.


\hide{
\section{\bf Proof of Proposition \ref{prop:simulation}} \label{app:proof simul}

In this appendix we provide the missing details of the proof of Proposition \ref{prop:simulation}. We use the same notation and definitions as in the partial proof provided in Section~\pur{\ref{subsec:net to fpcs}}. The key step is to modify the definition of the perturbation $U(\cdot)$ in \pur{\eqref{eq:def of U (not right cont) in the simulation lemma}}, so as to enforce right-continuity, which is done as follows. Instead of having $X(\cdot)$ jump to $y(\lceil t \rceil)$ immediately after the integer time $\lceil t \rceil$, we have this jump happen after a small delay.

For any integer time, we work with a vector $y(t)$ that has the properties described in Section~\pur{\ref{subsec:net to fpcs}}; cf. Claim~\ref{claim:simul discrete} \dela{and (???)}.
We define
\begin{equation}
\ptraj(t) =  \begin{cases} Q(\tf), &\textrm{if }  t-\tf < 2^{-\tf}  , \\
y(\tf), &\textrm{if }  t-\tf \ge 2^{-\tf},    \end{cases}
\end{equation}
and
\begin{equation}
\xi(t) =  \begin{cases} \big(\RM-I\big) \mu(\tf) +\lambda, & \textrm{if }  t-\tf < 2^{-\tf}, \\
 \big(\RM-I\big) \min\big( \mu(\tf),Q(\tf)\big)+\lambda, & \textrm{if }  t-\tf \ge 2^{-\tf}.    \end{cases}
\end{equation}
Once again, and as in the argument provided in Section~\ref{subsec:net to fpcs},
for any $t\in \Z_+$, we have $\xi(t)\in F\big(\ptraj(t)\big)+\lambda$.

We claim that $\ptraj(\cdot)$ is a perturbed trajectory of $F(\cdot)+\lambda$, corresponding to the right-continuous perturbation function $U(\cdot)$ defined below. As in Section~\pur{\ref{subsec:net to fpcs}}, summations are are carried out over non-negative integer values of the index $k$.
\begin{equation}
U(t) = \begin{cases}  \begin{array}{l}
\displaystyle{\sum_{k\le t-1}} \big( A(k)-\lambda \big) 
\, -\,  \tr \Big(\lambda+ \big(\RM-I\big) \mu(\tf) \Big) \\
\, +\, \displaystyle{\sum_{k\le t-1}} 2^{-k} \big(\RM-I\big) \Big( -\mu(k) + \min\big( \mu(k),Q(k)\big) \Big),
\end{array}
& \textrm{if } t-\tf < 2^{-\tf}, \vspace{10pt}\\
\begin{array}{l}
\displaystyle{\sum_{k< t-1}} \big( A(k)-\lambda \big) \,{-}\, \tr\lambda\,{-}\, 2^{-\tf}\big(\RM-I\big) \mu(\tf) \\
\, +\,  y(\tf)-Q(\tf) \,-\,  \big(t-\tf-2^{-\tf}\big) \big(\RM-I\big) \min\big( \mu(\tf),Q(\tf)\big)      \\
\, +\, \displaystyle{\sum_{k< t-1}} 2^{-k} \big(\RM-I\big) \Big( -\mu(k) + \min\big( \mu(k),Q(k)\big) \Big),
\end{array}
& \textrm{if } t-\tf \ge 2^{-\tf}.
\end{cases} 
\end{equation}
For any $k\in \Z_+$, we have
\begin{equation}\label{eq:int = ptraj 1 app}
\begin{split}
\int_k^{k+1} \xi(\tau)\,d\tau \,+\, \Prt(k+1)- \Prt(k) \, &= \,
2^{-k} \big(\RM-I\big) \mu(k) \, +\, \big(1-2^{-k}\big)\big(\RM-I\big) \min\big( \mu(k),Q(k)\big) + \lambda\\
&\quad  +\big(A(k)-\lambda\big) + 2^{-k} \big(\RM-I\big) \Big( -\mu(k) + \min\big( \mu(k),Q(k)\big) \Big)\\
&=\, A(k) \,+\,  \big(\RM-I\big)\min\big( \mu(k),Q(k)\big)\\
&=\, Q(k+1) - Q(k)\\
&=\, X(k+1)-X(k).
\end{split}
\end{equation}
Furthermore, for any $t$ such that $ t-\tf < 2^{-\tf}$,
\begin{equation}\label{eq:int = ptraj 2 app}
\begin{split}
\int_{\tf}^{t} \xi(\tau)\,d\tau \,+\, \Prt(t)- \Prt(\tf) \, &= \,
\tr \Big(\lambda + \big(\RM-I\big) \mu(\tf)  \Big) - \tr \Big(\lambda + \big(\RM-I\big) \mu(\tf)  \Big)\\
&=\,0\\
&=\, X(t)-X(\tf).
\end{split}
\end{equation}
Finally, for any $t$ such that $ t-\tf \ge 2^{-\tf}$,
\comm{In the third line== below, the signs do not seem to agree with the signs in the definition of $U(t)$. I think that some of the signs in the definition of $U$ have to change, namely the first line for the second case needs to be changed to
$\red{-}\, \tr\lambda\, \red{-}\, 2^{-\tf}\big(\RM-I\big) \mu(\tf)$. Agree?} \oli{Yes. Fixed.}
\begin{equation}\label{eq:int = ptraj 3 app}
\begin{split}
\int_{\tf}^{t} \xi(\tau)\,d\tau \,+\, \Prt(t)- \Prt(\tf) \, &= \,
\bigg[\tr\lambda \,+\, 2^{-\tf} \big(\RM-I\big) \mu(\tf) \\
&\phantom{=  \,\bigg[}+\, \big(t-\tf-2^{-\tf}\big) \big(\RM-I\big) \min\big( \mu(\tf),Q(\tf)\big)\bigg]  \\
&\phantom{= \, } +\bigg[- \tr\lambda  \,-\, 2^{-\tf} \big(\RM-I\big) \mu(\tf) \\
&\phantom{=  \,\bigg[}  -\,  \big(t-\tf-2^{-\tf}\big) \big(\RM-I\big) \min\big( \mu(\tf),Q(\tf)\big) \\
&\phantom{=  \,\bigg[}+\, y(\tf)-Q(\tf)\bigg]\\
&=\,y(\tf) - Q(\tf)\\
&=\, X(t)-X(\tf).
\end{split}
\end{equation}
An induction on \eqref{eq:int = ptraj 1 app}, \eqref{eq:int = ptraj 2 app}, and  \eqref{eq:int = ptraj 3 app} shows that  \eqref{eq:def of integral pert 1}, i.e., $\ptraj(t) = \int_0^t \xi(\tau)\,d\tau \,+\, \Prt(t)$, holds for all times $t$.
Thus, $\ptraj(\cdot)$ is indeed a perturbed trajectory of $F(\cdot)+\lambda$ corresponding to $\Prt(\cdot)$.

We will now verify the perturbation bound \eqref{eq: U bounded cond}. Recall the constant $b$ in \eqref{eq:def constant c}. For any $t\in\R_+$,
\begin{equation}
\begin{split}
\sup_{\tau\le t} \Ltwo{\Prt(\tau)}\, & \le \, 
 \sup_{\tau\le t} \Ltwo{ \sum_{k < \tau } \big(A(k) -\lambda   \big) } \,+\, \sup_{k\in Z_+} \Ltwo{y(k)-Q(k) }  \,+\, \lVert  \lambda \rVert \,  \\
& \quad  +\, \sup_{k < \tau } \Ltwo{ \big(\RM-I\big) \mu(\tf)} \, 
 + \sup_{k < \tau} \Ltwo{ \big(\RM-I\big) \min\big( \mu(\tf),Q(\tf)\big)}   \\
 &\quad  +\, \sum_{k < \tau} 2^{-k} \Big( \Ltwo{\big(\RM-I\big) \mu(k)}  +  \Ltwo{ \big(\RM-I\big) \min\big( \mu(k),Q(k)\big)} \Big)\\
&\le\,  \sup_{\tau\le t} \Ltwo{ \sum_{k < \tau} \big(A(k) -\lambda   \big) } \, +\, \kappa \,+\, \lVert  \lambda \rVert \,+\, b \,+\,b\, +\,  \sum_{k= 0}^\infty 2^{-k}\big(b+b\big) \\
&=\,  \sup_{\tau\le t} \Ltwo{ \sum_{k < \tau} \big(A(k) -\lambda   \big) }\,+\, \lVert  \lambda \rVert\, + \, \kappa \,+\, 6b.
\end{split}
\end{equation}	
This implies \eqref{eq: U bounded cond} for $\beta=\kappa + 6b$, and completes the proof of Proposition \ref{prop:simulation}.

\medskip

}


\hide{
\section{\bf SSC in No Super Exponential Scaling} \label{app:ssc iid example}
In this appendix we show by an example that no SSC  occurs under any super exponential scaling of time. This will serve as a converse for Corollary \ref{prop:ssc iid trans}.

\oli{Is it too obvious? Remove this appendix?}

\red{Yes. Remove the appendix. Can make a comment in the main text. But will have to be a bit more precise. Two possible versions:
\\
superexponential: $g(r)$ grows faster than any exponential
\\
large exponential: $g(r)=e^{\gamma r}$, where $\gamma$ is sufficiently large, as in the example below.
\\
Pick whichever version will lead to the shortest comment.}

Consider a network of one queue and one server with unit service rate. Let  $\lambda=1/2$, and let $A(t)$ be  an iid arrival process with
\begin{equation}
A(t) = \begin{cases} 2, & \textrm{with prob } 1/4, \\ 0, & \textrm{with prob } 3/4.  \end{cases}
\end{equation}
The set of invariant states is then a singleton, $\I(\lambda) = \{0\}$. Fix some $\delta>0$ and let $g(r)=\exp\big(1.5 r\delta\big)$. Consider a sequence of $g$-time-scaled trajectories:
\begin{equation}
\qhat^r(t) = \frac{1}{r} Q\big(e^{1.5 r\delta}t\big),
\end{equation}
for $t\le T$.
We will show that 
\begin{equation}\label{eq:no ssc with big exp}
\prob{\sup_{t \le T}\,d\Big(\qhat^r(t)\,,\, \I\Big)>\delta}  \xrightarrow[{\,r\to\infty\,}]{} 1.
\end{equation}
As a result, SSC does not occur for $\exp\big(1.5 r\delta\big)$ scaling of time. 

We divide the $\exp\big(1.5 r\delta\big) T$ time window into intervals of length $r\delta$, and associate an event $E_i$ to the $i$th interval:
\begin{equation}
E_i\triangleq \textrm{ The event that } A(t)=2, \textrm{ for all }  t\in \big[(i-1)r\delta,\, ir\delta \big).
\end{equation}
These events are independent and have probability $\prob{E_i}=\big(1/4\big)^{r\delta}$. If $E_i$ occurs, then $Q\big(ir\delta\big)> r\delta$, and $\qhat^r\big(r\delta/g(r)\big)> \delta$. Hence,
\begin{equation}
\begin{split}
\prob{\sup_{t \le T}\,d\Big(\qhat^r(t)\,,\, \I\Big)\ge\delta}  &\ge \prob{\bigcup_{i\le e^{1.5 r\delta}T/r\delta} E_i} \\
& = 1-\prob{\bigcap_{i\le e^{1.5 r\delta}T/r\delta} E_i^c}\\
& = 1- \probb{E_1^c}^ {e^{1.5 r\delta}T/r\delta}\\
& = 1- \left(1-4^{-r\delta}\right)^ {e^{1.5 r\delta}T/r\delta}\\
& = 1- \exp\left( -4^{-r\delta}  {e^{1.5 r\delta}T/r\delta}\right)\\
&= 1- \exp\left(-e^{(1.5-\ln 4) r\delta}T/r\delta\right)\\
&\xrightarrow[{\,r\to\infty\,}]{} 1.
\end{split}
\end{equation}
where the last inequality is because $1.5>\ln 4$. 

\oli{We must be able to show a more general result: For any network with any iid arrival process, SSC occurs for no scaling of time larger than $\exp(\gamma\sigma)$, where $\gamma$ is a (universal) constant and $\sigma$ is the standard deviation of the arrival (projected on $\I^\perp$). Right? Good to mention?}
}


\bibliography{schbib}
\bibliographystyle{imsart-number}

\end{document}